\documentclass[12pt,transaction, draftclsnofoot,onecolumn]{IEEEtran}
\usepackage{amsfonts}
\usepackage{amssymb}
\usepackage{hyperref}
\usepackage{graphicx}
\usepackage{subfigure}
\usepackage{enumerate}
\usepackage{amsmath}
\usepackage{color}
\usepackage{amsthm}
\usepackage{amsmath}
\usepackage{algorithm,algorithmic}
\usepackage{color}
\newcommand{\bp}{\begin{proof} \small }
\newcommand{\ep}{\end{proof} \normalsize}
\newcommand{\epx}{\end{proof} \small}
\newcommand{\bpa}{\begin{proofappx} \footnotesize }
\newcommand{\epa}{\end{proofappx} \small }

\newtheorem{proposition}{Proposition}

\newtheorem{lemma}{Lemma}

\newtheorem{definition}{Definition}

\newtheorem*{theorem*}{Theorem}
\newtheorem*{proposition*}{Proposition}
\newtheorem*{corollary*}{Corollary}
\newtheorem*{lemma*}{Lemma}
\newtheorem*{assumption*}{Assumption}
\newtheorem*{definition*}{Definition}
\newtheorem*{claim*}{Claim}

\newcommand{\be}{\begin{equation}}
\newcommand{\ee}{\end{equation}}
\newcommand{\bs}{\begin{subequations}}
\newcommand{\es}{\end{subequations}}
\newcommand{\bq}{\begin{eqnarray}}
\newcommand{\eq}{\end{eqnarray}}
\newcommand{\bqn}{\begin{eqnarray*}}
\newcommand{\eqn}{\end{eqnarray*}}

\newcommand{\ba}{\left[ \begin{array}}
\newcommand{\ea}{\\ \end{array} \right]}
\newcommand{\ben}{\begin{enumerate}}
\newcommand{\een}{\end{enumerate}}

\def\a{{\boldsymbol{a}}}

\def\d{{\boldsymbol{d}}}

\def\real{{\mathchoice%
{\hbox{\rm\setbox1=\hbox{I}\copy1\kern-.45\wd1 R}}
{\hbox{\rm\setbox1=\hbox{I}\copy1\kern-.45\wd1 R}}
{\hbox{\scriptsize\rm\setbox1=\hbox{I}\copy1\kern-.45\wd1 R}}
{\hbox{\scriptsize\rm\setbox1=\hbox{I}\copy1\kern-.45\wd1 R}}}}

\def\Zint{{\mathchoice{\setbox1=\hbox{\sf Z}\copy1\kern-.75\wd1\box1}
{\setbox1=\hbox{\sf Z}\copy1\kern-.75\wd1\box1}
{\setbox1=\hbox{\scriptsize\sf Z}\copy1\kern-.75\wd1\box1}
{\setbox1=\hbox{\scriptsize\sf Z}\copy1\kern-.75\wd1\box1}}}
\newcommand{\complex}{ \hbox{\rm C\kern-0.45em\rule[.07em]{.02em}{.58em}%
\kern 0.43em}}

\begin{document}
\title{Collaborative Service Caching for Edge Computing in Dense Small Cell Networks}
\author{Lixing~Chen,~\IEEEmembership{Student~Member,~IEEE},
        ~Jie Xu,~\IEEEmembership{Member,~IEEE}% <-this % stops a space
\thanks{L. Chen and J. Xu are with the Department of Electrical and
	Computer Engineering, University of Miami. Email: lx.chen@miami.edu, jiexu@miami.edu}}
\maketitle
\vspace{-0.3 in}
\begin{abstract}
%\boldmath
Mobile Edge Computing (MEC) pushes computing functionalities away from the centralized cloud to the proximity of data sources, thereby reducing service provision latency and saving backhaul network bandwidth. Although computation offloading has been extensively studied in the literature, service caching is an equally, if not more, important design topic of MEC, yet receives much less attention. Service caching refers to caching application services and their related data (libraries/databases) in the edge server, e.g. MEC-enabled Base Station (BS), enabling corresponding computation tasks to be executed. Since only a small number of services can be cached in resource-limited edge server at the same time, which services to cache has to be judiciously decided to maximize the system performance. In this paper, we investigate collaborative service caching in MEC-enabled dense small cell (SC) networks. We propose an efficient decentralized algorithm, called CSC (Collaborative Service Caching), where a network of small cell BSs optimize service caching collaboratively to address a number of key challenges in MEC systems, including service heterogeneity, spatial demand coupling, and decentralized coordination. Our algorithm is developed based on parallel Gibbs sampling by exploiting the special structure of the considered problem using graphing coloring. The algorithm significantly improves the time efficiency compared to conventional Gibbs sampling, yet guarantees provable convergence and optimality. CSC is further extended to the SC network with selfish BSs, where a coalitional game is formulated to incentivize collaboration. A coalition formation algorithm is developed by employing the \textit{merge-and-split} rules and ensures the stability of the SC coalitions. Systematic simulations are carried out to evaluate the efficacy and performance of the proposed algorithm. The results show that our algorithm can effectively reduce edge system operational cost for both cooperative and selfish SCs.
\end{abstract}

%\begin{IEEEkeywords}
%IEEEtran, journal, \LaTeX, paper, template.
%\end{IEEEkeywords}

% \fi
\IEEEpeerreviewmaketitle

\section{Introduction}
Modern mobile applications, such as Virtual/Augmented Reality (VR/AR), Immersible Communications, Mobile Gaming, and Connected Vehicles, are becoming increasingly data-hungry, compute-demanding and latency-sensitive. While cloud computing has been leveraged to deliver elastic computing power and storage to support resource-constrained end-user devices in the last decade, it is meeting a serious bottleneck as moving all the distributed data and computing-intensive applications to the remote clouds not only poses an extremely heavy burden on today's already-congested backbone networks but also results in large transmission latencies that degrade quality of service. As a remedy to these limitations, mobile edge computing (MEC) (a.k.a. fog computing) \cite{mao2017mobile,shi2016edge} has recently emerged as a new computing paradigm to push the frontier of data and services away from centralized cloud infrastructures to the logical edge of the network, thereby enabling analytics and knowledge generation to occur closer to the data sources.

MEC provides cloud computing capabilities within the radio access network in the vicinity of mobile subscribers. As a major deployment scenario, edge servers are deployed at mobile base stations (BSs) to serve end-users' computation tasks \cite{mao2016dynamic}, thereby mitigating the congestion within the core network. On the other hand, BSs are also able to access the cloud data centers in case of a need for higher computation power or storage capacity, resulting in a hierarchical computation offloading architecture among end-users, BSs, and the cloud. While computation offloading has been the central theme of most works studying MEC, or more broadly, mobile cloud computing (MCC), what is often ignored is the heterogeneity and diversity of mobile services, and how these services are cached on BSs in the first place. Specifically, service caching (or service placement) refers to caching computation services and their related libraries/databases in the edge server co-located with the BS based on spatial-temporal service popularity information \cite{mao2017mobile}, thereby enabling user applications requiring these services to be executed at the network edge in a timely manner. For instance, visitors in a museum can use AR service for better sensational experience. Thus, it is desirable to cache AR services at the BS serving this region for providing the real-time service. On the other hand, users  play mobile games at home in the evening. Such information will then suggest operators to cache gaming services during this typical period for handling huge computation loads.

However, unlike the cloud which has huge and diverse resources, the limited computing and storage resources of an individual BS allow only a very small set of services to be cached at a time. As a result, which services are cached on the BS determines which application tasks can be offloaded to the edge server, thereby significantly affecting the edge computing performance. Although the BS can adjust its cached service on the per-task basis, the heterogeneous and time-varying user tasks, especially when the BS is serving multiple users, would lead to frequent \textit{cache-and-tear} operations \cite{mao2017mobile}. This not only imposes additional traffic burden on the core network but also causes extra computation latency, thereby degrading user quality of experience. As such, service caching is a relatively long-term decision compared to computation offloading, and which services to cache must be judiciously decided given the limited resources of individual BSs based on the predicted service popularity in order to optimize the edge computing performance.

Parallel to MEC, the Small Cell (SC) concept has appeared to provide seamless coverage, boost throughput and save energy consumption in cellular networks. The density of BSs in cellular networks has kept increasing since its birth to nowadays 4G networks, and is expected to reach about 40-50 BSs/km$^2$ in next generation 5G networks and beyond, namely the ultra-dense networks (UDN) \cite{ge20165g}. It is envisioned that the majority of SCs will be purchased, deployed and maintained by either end-users in their homes in a ``plug and play'' fashion, or enterprises in commercial buildings with no or minimal assistance from mobile operators. By enabling these SCs to operate in an open subscriber group (OSG) mode \cite{garcia2010open} and thus provide service in their immediate vicinities, mobile operators can significantly lower their capital expenditures and operational expenditures. In such dense networks, a typical mobile device will be within the coverage of multiple SCs. This creates an opportunity for nearby small cell BSs to virtually pool their computation resources together, share their cached services and collaboratively serve users' computation workload. Specifically, small-cell BSs can collaboratively decide which services to cache so that the set of services available at the network edge is maximized to the users, see Fig.\ref{scenario} for illustration.

\begin{figure}[htb]
	\centering
	\includegraphics[width=0.5\textwidth]{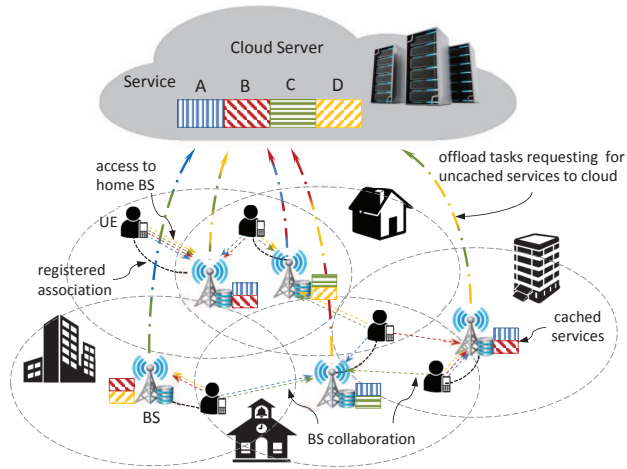}\\
	\vspace{-0.1 in}
	\caption{System illustration. A BS can only cache a subset of services; requested service can be executed by the BS only if it is cached at that BS; requests for services not cached at registered BS can be satisfied by BS collaboration or cloud offloading; BSs have to jointly optimize service caching decisions.}\label{scenario}
    \vspace{-0.3 in}
\end{figure}

Although the idea of collaborative service caching in dense small cell networks offers a promising solution to enhance MEC performance, how to achieve optimal collaborative service caching faces many challenges. Firstly, compared to conventional clouds that manage only the computing resource, MEC leads to the co-provisioning of radio access and computing services by the BSs, thus mandating a new model that can holistically capture both aspects of the system. BSs can collaborate not only in shared service caching but also the physical transmission for service delivery. Secondly, although dense deployment creates much coverage overlapping, BSs are still geographically distributed. This leads to a complex multi-cell network where demand and resources become highly coupled and hence an extremely difficult combinational optimization problem for collaborative service caching. Thirdly, since small cells are often owned and deployed by selfish individual users (e.g. home/enterprise owners), small cells will be reluctant to participate in the collaborative system without proper incentives. Therefore, incentives must be devised and incorporated into the collaborative service caching scheme.

In this paper, we study the extremely compelling but must less investigated problem of service caching in MEC-enabled dense small cell networks, and develop novel techniques for performing collaborative service caching in order to optimize the MEC performance. The main contributions of this paper are summarized as follows.

(1) We formulate the collaborative service caching in MEC-enabled dense small cell networks as a utility maximization problem. Both cases where small cells are fully cooperative and strategic (i.e. self-interested) are considered. To our best knowledge, this is the first work that studies collaborative service caching in dense networks.

(2) We develop an efficient decentralized algorithm, called CSC (Collaborative Service Caching), to solve the collaborative service caching problem in cooperative SC networks, which is an extremely difficult combinational optimization problem. Our algorithm is developed based on the Gibbs sampler, a popular Markov Chain Monte Carlo (MCMC) algorithm. Innovations are made to enable parallel Gibbs sampling based on the graph coloring theory, thereby significantly accelerating the algorithm convergence.

(3) We employ coalitional game to understand the behavior of strategic SCs when performing collaborative service caching. Incentive mechanisms are used to promote SC collaboration within generated coalitions. A decentralized coalition formation algorithm based on merge-and-split processes is developed and we prove that the proposed algorithm results in \textit{Pareto-optimal} stable coalitions among the SCs.

The rest of this paper is organized as follows. Section II reviews related works. Section III presents the system model. Section IV focuses on the collaborative service caching problem with fully cooperative SCs. Section V designs the collaborative service caching for strategic SCs. Section VI presents the simulation results, followed by the conclusion in Section VII.

\section{Related Work}
Computation offloading is the central theme of many prior studies in both mobile cloud computing (MCC) \cite{fernando2013mobile,buyya2009cloud} and mobile edge computing (MEC) \cite{mao2017mobile,shi2016edge}, which concerns what/when/how to offload users' workload from their devices to the cloud or the edge servers. Various works have studied different facets of this problem, considering e.g. stochastic task arrivals \cite{huang2012dynamic,liu2016delay}, energy efficiency \cite{xu2017online,mao2016dynamic}, collaborative offloading \cite{chen2017socially,tanzil2016distributed} etc. However, the implicit assumption is that edge servers can process whatever types of computation tasks that are offloaded from users without considering the availability of services in edge servers, which in fact is crucial in MEC due to the limited resources of edge servers.

Similar service caching/placement problems, known as virtual machine (VM) placement, have been investigated in conventional cloud computing systems. VM placement over multiple clouds is studied in \cite{tordsson2012cloud,li2009enacloud,gao2013multi}, where the goal is to reduce the deployment cost, maximize energy saving and improve user experience, given constraints on hardware configuration and load balancing. However, these works cannot be directly applied to design efficient service caching policies for MEC since mobile networks are much more complex and volatile, and optimization decisions are coupled both spatially and temporally.

Service caching/placement is also related to content caching/placement in network edge devices \cite{shanmugam2013femtocaching}. Various content caching strategies have been proposed, e.g., the authors in \cite{prabh2005energy} aim to find optimal locations to cache the data that minimize packet transmissions in wireless sensor nodes and the authors in \cite{liu2017caching} employ zero-forcing beamforming to null interference and optimize the content caching for maximizing the success probability. The concept of FemtoCaching is introduced in \cite{shanmugam2013femtocaching} which studies content placement in small cell networks to minimize content access delay. The idea of using caching to support mobility has been investigated in \cite{wang2015dynamic}, where the goal is to reduce latency experienced by users moving between cells. Learning-based content caching policies are developed in \cite{muller2017context} for wireless networks with \textit{a priori} unknown content popularity. While content caching mainly concerns with storage capacity constraints, service caching has to take into account both computing and storage constraints, and has to be jointly optimized with task offloading to maximize the overall system performance.

Service caching/placement for edge systems has gained little attention until very recently. \cite{skarlat2017towards} presents a centralized service placement strategy which allows placement of IoT services on virtualized fog resources while considering QoS constraints like execution deadline. The service placement problem in \cite{skarlat2017towards} is formulated to be an integer linear programming problem and solved by CPLEX \cite{CPLEX} in a centralized manner. By contrast, the collaborative service caching problem in our paper is a non-linear combinational problem, which is much more difficult to be solved by existing Non-linear Integer Programming solvers. More importantly, instead of employing a centralized controller, we are more interested in a distributed network setting where BSs are able to make caching decisions in a decentralized manner without collecting the information of the whole network at a controller node. The most related work probably is \cite{yang2016cost}, which studied the joint optimization of service caching/placement over multiple cloudlets and load dispatching for end users' requests. There are several significant differences of our work. First, while the coverage areas are assumed to be non-overlapping for different cloudlets in \cite{yang2016cost}, BSs have overlapping coverage areas in our considered dense cellular network. Second, while \cite{yang2016cost} only develops heuristic solutions for service caching problem, we firmly prove the optimality of our algorithm. Third, while the algorithm is centralized in \cite{yang2016cost}, our algorithm enables decentralized coordination among BSs. Moreover, we consider the selfish nature of small cell owners. We also note that the term ``service placement'' was used in some other literature \cite{wang2017dynamic,zhang2013dynamic} in a different context. The concern is to assign task instances to different clouds but there is no limitation on what types of services that each cloud can run.

\section{System Model}
\subsection{Network}
We consider a network of $N$ densely deployed SCs, indexed by $\mathcal{N}$. Each SC $n \in \mathcal{N}$ has an access point/base station (BS) endowed with computation and storage resources and hence can provide edge computing services to end users in its radio range. Each SC $n$ has a set of registered user equipments (UEs) within its coverage, indexed by $\mathcal{M}_n$. Let the set of all users be $\mathcal{M} = \cup_{n \in\mathcal{N}} \mathcal{M}_n$. For instance, in a typical IoT scenario in a multi-story building, a small cell BS (e.g. femtocell BS) and multiple mobile devices, sensors and things are deployed in each room of this building. UEs are authorized to access the edge computing resources of their home BS (i.e. the BS of the SC that they are registered to). However, due to the dense deployment of SCs in the building, a UE is in the radio coverage of multiple SCs besides its home BS. For each UE $m$, let $\mathcal{N}_m \subseteq \mathcal{N}$ be the set of reachable BSs. Let $H_{m,n}$ be the uplink channel condition between UE $m$ and BS $n$. Typically, the home BS of a UE has the best channel condition since it is often in the same room of the UE.

We say that two SCs $i, j\in \mathcal{N}$ are neighbors if there exists some UE $m$ such that $i, j \in \mathcal{N}_m$. In other words, SCs $i$ and $j$ can potentially collaborate with each other to serve at least one common UE. Based on this definition, the network of SCs can be described by a graph $G = \langle\mathcal{N}, \mathcal{E}\rangle$ where $\mathcal{E}$ is the edge set and there exists an edge between SCs $i$ and $j$ if they are neighbors. Let $\Omega_i \subseteq \mathcal{N}$ denote the one-hop neighborhood of SC $i$.

\subsection{Services and Demand}
Service is an abstraction of application that is hosted by the BS and requested by UEs. Example services include video streaming, mobile gaming and augmented reality. Running a particular service requires caching the associated data, such as required libraries and databases, on the BS. We assume that there are totally $K$ possible computing services, indexed by $\mathcal{K}$. Different services require different processor power, memory, and storage resources. However, due to the limited computing resources on an individual BS, not all $K$ services can be cached at the same time. Given the computing resources (on various dimensions) of SC $n$, the feasible sets of services that can be cached on SC $n$ simultaneously can be easily derived. Let $\mathcal{F}_n$ be the set of feasible sets of services for SC $n$ and each element in $\mathcal{F}_n$ is a feasible set of services. As a simple example, consider that each SC can only offer one kind of service at a time and hence, $\mathcal{F}_n = \{\{1\}, \{2\}, ..., \{K\}\}$ for all $n$ in this case.

The network of SCs periodically update their service caching decisions according to the predicted/reported service demand of the UEs. We focus on the decision problem for one such period, which is much longer than the timescale of computation offloading (i.e. packet transmission + task processing). Let $a_n \in \mathcal{F}_n$ be the caching decision of SC $n$ and $\a = (a_1, ..., a_N)$ collects the caching decisions of all SCs. The service demand of UE $m$ in the current period is denoted by a vector $\d_m = (d_m^1, ..., d_m^K)$ where $d_m^k = (\lambda_m^k, \gamma_m^k)$ is a tuple capturing the input data size $\lambda_m^k$ (in bits) and the computation workload $\gamma_m^k$ (in processor cycles) of demand requiring service $k$.

\subsection{Utility Model}
In SC networks, UEs derive benefits from completing their computation tasks. Let $u^k_m$ be the benefit that UE $m$ receives from tasks requiring service $k$. Meanwhile, costs are incurred in order to complete these tasks, which come from two major sources: cost due to transmitting the data of these tasks from UEs to BSs and cost due to processing the computation workload of these tasks. We model them in more detail as follows.

\subsubsection{Transmission cost}
Let $P_m$ denote the transmission power of UE $m$, then the achievable uplink transmission rate between UE $m$ and BS $n$ can be calculated using the Shannon capacity
\begin{align}
r_{m,n} = W\log_2\left(1 + \frac{P_m H_{m,n}}{N_0}\right)
\end{align}
where $W$ is the channel bandwidth and $N_0$ is the noise power. Therefore, the transmission energy consumption of UE $m$ for sending $\lambda_{m,n}$ bits of input data to BS $n$ is $P_m \lambda_{m,n}/r_{m,n}$. The total energy consumption of UE $m$ is thus
\begin{align}\label{TXcost}
C^{\text{tx}}_m = P_m \sum_{n \in \mathcal{N}_m} \lambda_{m,n}/r_{m,n}
\end{align}
where $\lambda_{m,n}$ is the amount of data sent to BS $n$ from UE $m$, which is jointly determined by the SC service caching decisions and UE-SC association as will be elaborated later in this paper. Notice that we neglect the transmission cost for BSs to send service outcome back to the UEs due to the fact that the size of service outcome is much smaller than that of service input data. 

\subsubsection{Computation cost} \label{comcost}
If BS $n$ caches service $k$, then the received workload requiring service $k$ can be processed at BS $n$, thus incurring computation cost, e.g. energy consumption; otherwise, the computation is further offloaded to the cloud, thus incurring cloud usage cost, e.g., cloud service fee and Internet round-trip delay. Let $c_n$ be the unit computation consumption of BS $n$ and $c_0$ be the unit cloud usage cost. We assume that $c_0 > c_n, \forall n$ so that processing at the network edge is preferred compared to offloading to the remote cloud. Let $\gamma_{m,n}, n \in\mathcal{N}$ and $\gamma_{m,0}$ denote the workload processed for UE $m$ by BS $n$ and cloud, respectively, which are related to the SCs' service caching decision and UE-SC association. Therefore, we have the computation cost of SC $n$, $C^{\text{com}}_{m,n}$, and cloud, $C^{\text{com}}_{m,0}$, for executing the service demand from UE $m$:
\begin{align}\label{PXcost}
C^{\text{com}}_{m,n} = c_n \gamma_{m,n}, \forall n\in\mathcal{N} \quad \text{and} \quad C^{\text{com}}_{m,0} = c_0 \gamma_{m,0}
\end{align}

%We further write $v^k_m\in\mathcal{V}\triangleq\{\mathcal{N}\cup 0\}$ as the serving SC/cloud for UE $m$'s demand $d^k_m$, therefore, the variable $\gamma_{m,i}$ can be obtained as follows:
%\begin{align}
%\gamma_{m,i} = \sum_k \gamma^k_m \cdot \textbf{1}\{v^k_m=i\}
%\end{align}
%where $\mathbf{1}\{\cdot\}$ is the indicator function and $\gamma^k_m$ is the amount of workload requiring service $k$ from UE $m$.

% \subsubsection{Payment scheme}\label{PWR}

We consider a \textit{pay-to-work} service provision framework where SCs are paid by UEs to process their service requests. Without loss of generality, we employ an intuitive payment scheme: if a SC or cloud offers services to UE $m$, then UE $m$ must pay the SC or cloud to cover the incurred computation cost. Let $C^{\text{com}}_{m}$ be the total payment of UE $m$, we have
\begin{align}
C^{\text{com}}_{m} = \sum_{n\in\mathcal{N}} C^{\text{com}}_{m,n} + C^{\text{com}}_{m,0}
\end{align}

The utility of SC $n$ is obtained by accumulating the differences between benefits and costs (i.e. transmission cost and payments) of its registered UEs. Recall that $u^k_m$ is the benefit UE $m$ receives if the request for service $k$ is satisfied, then utility of SC $n$ can be represented as:
\begin{align}
U_n = \sum_{m\in\mathcal{M}_n} \left(\sum_k u^k_m - C^{\text{com}}_{m} - C^{\text{tx}}_m\right)
\end{align}

\subsection{Non-Collaborative Service Caching}
Before we study collaborative service caching, we first present the non-collaborative service caching problem as a benchmark. In this case, UEs transmit all service request to their home BSs. Therefore, if $n$ is UE $m$'s home BS, $\lambda_{m,n} = \sum_{k}\lambda^k_m$; otherwise, $\lambda_{m,n} = 0$. Moreover, BS $n$ only processes tasks requesting for services cached locally and the remaining tasks will be further offloaded to the cloud server, hence we have $\gamma_{m,n}=\sum_{m\in\mathcal{M}_n}\gamma^k_m \cdot \textbf{1}\{k \in a_n\}$ and $\gamma_{m,0}=\sum_{m\in\mathcal{M}_n}\gamma^k_m \cdot \textbf{1}\{k \not\in a_n\}$. Depending what services $a_n\in\mathcal{F}_n$ are cached on BS $n$, the utility of SC $n$ can be written as:
%\begin{align}
%U_n(a_n) = \sum_{m\in\mathcal{M}_n} & \left( \sum_k u^k_m - P_m \sum_{k}\lambda^k_{m,n}/r_{m,n}\right. \nonumber \\ & \left. - c_n\sum_{k} \gamma^k_m \cdot \textbf{1}\{k\in a_n\} -  c_0\sum_{k} \gamma^k_m \cdot \textbf{1}\{k \not\in a_n\} \right)
%\end{align}
\begin{align*}
U_n(a_n) = \sum_{m\in\mathcal{M}_n} \left( \sum_k u^k_m - P_m \sum_{k}\lambda^k_{m}/r_{m,n} - c_n\sum_{k} \gamma^k_m \cdot \textbf{1}\{k\in a_n\} -  c_0\sum_{k} \gamma^k_m \cdot \textbf{1}\{k \not\in a_n\} \right)
\end{align*}
Notice that $\textbf{1}\{k\in a_n\}+\textbf{1}\{k \not\in a_n\} \equiv 1$, the utility function of SC $n$ can be rewritten as:
\begin{align}\label{utility_rew}
U_n(a_n) = \sum_{m\in\mathcal{M}_n} \left( \sum_k u^k_m - P_m \sum_{k}\lambda^k_{m}/r_{m,n} - (c_n-c_0) \sum_{k} \gamma^k_m \cdot \textbf{1}\{k\in a_n\} - c_0 \sum_{k} \gamma^k_m \right) 
\end{align}

Each SC $n$ aims to maximize its utility by optimizing its service caching decision, i.e. $\min_{a_n \in \mathcal{F}_n} U_n(a_n)$. Realizing that the UEs' benefit and transmission cost do not depend on the service caching decision in the non-collaborative case. $(c_n-c_0) \sum_{k} \gamma^k_m \cdot \textbf{1}\{k\in a_n\}$ is the only term that depends on the service caching decision in \eqref{utility_rew}. Recall that $c_0>c_n$, the optimization problem can be further simplified to
\begin{align}
\max_{a_n \in \mathcal{F}_n}~~ \sum_{m\in\mathcal{M}_n} \sum_k \gamma^k_m\cdot \mathbf{1}\{k \in a_n\}
\end{align}

The above problem is not difficult to solve: the optimal solution requires the SC to cache the most popular services (i.e. services with the largest $\sum_{m\in\mathcal{M}_n}\gamma^k_m$). However, non-collaborative SCs may rely heavily on the cloud service, which incurs large latency and cloud usage fee, due to the diversity of user demand and the limited caching capacity of BSs. One promising way to reduce the reliance on cloud is to enable collaboration among SCs. Figure \ref{noncop_cop} illustrates a simple motivating example for collaborative service caching. Consider two overlapping SCs providing two services to their respective UEs, and assume that each SC can only cache one service. Without collaboration (the left figure), each SC caches the most popular service, which is the red service in this case, in order to maximize its own utility. All computation demand for the green service will have to be offloaded to the cloud. With collaboration (the right figure), SC 1 caches the red service while SC 2 caches the green service. Because the coverage areas of these two SCs are highly overlapped, computation demand for the green service from UEs registered to SC 1 can be offloaded to SC 2. Likewise, computation demand for the red service from UEs registered to SC 2 can be offloaded to SC 1. In this way, all computation is retained at the network edge with zero cloud usage. In the next sections, we will study how to perform collaborative service caching to improve the MEC performance. We will first assume that SCs are fully cooperative with the common goal of maximizing the MEC performance of the whole network. Next, we study the problem considering selfish SCs.

\begin{figure}[htb]
	\centering
	\includegraphics[width=0.7\textwidth]{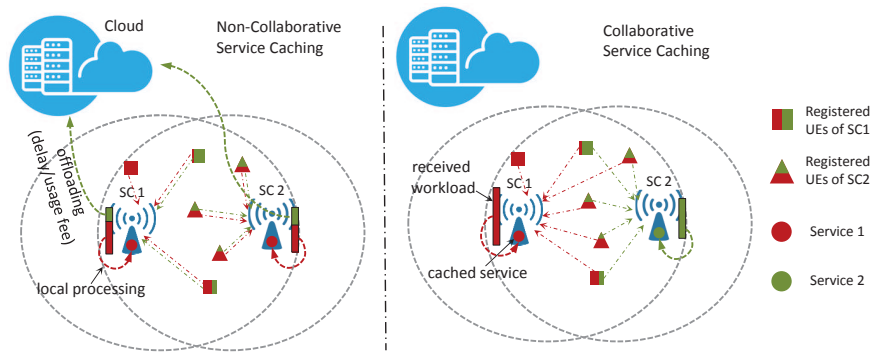}\\
	\vspace{-0.15 in}
	\caption{Comparison of non-collaborative service caching and collaborative service caching}\label{noncop_cop}
	\vspace{-0.1 in}
\end{figure}

\section{Collaborative Service Caching for Obedient SCs}\label{sec_CSC}
In this section, we study collaborative service caching assuming that SCs are obedient and fully cooperative, who have the common goal of maximizing the total utility of the whole network. With collaborative service caching, a UE does not have to send its computation workload to its home BS. Instead, the service caching decisions of the SCs will reshape the workload distribution in the network in order to better utilize the limited computing resources of individual BSs.

\subsection{Problem Formulation}
Given the service caching decisions $\a$ of all SCs, let $\mathcal{B}^k_m(\a) \subseteq \mathcal{N}$ be the set of the SCs that cache service $k$ and are reachable by UE $m$. Since UEs are often battery-powered and have stringent energy constraints, in our considered cooperative system, the demand $d^k_m$ of UE $m$ is always offloaded to the SC in $\mathcal{B}^k_m(\a)$ that has the best uplink channel condition, namely $\arg \max_{n\in\mathcal{B}^k_m(\a)} H_{m,n}$. In this way, the UE incurs the least transmission energy consumption. Notice that our system is also compatible with other UE-SC association strategies. Designing appropriate association strategy helps improve the performance of MEC system, however, this goes beyond the main purpose of our analysis. Interested readers are referred to \cite{urgaonkar2015dynamic} and references therein. It is still possible that none of the SCs that UE $m$ can reach caches service $k$, namely $\mathcal{B}^k_m(\a) = \emptyset$. In this case, UE $m$'s demand $d^k_m$ will be offloaded to the cloud via UE $m$'s home BS. To facilitate the exposition, we write $v^k_m(\a)\in \mathcal{N}$ as the BS to which UE $m$ send its demand for service $k$, $d^k_m$, which can be determined as follows:
\begin{equation}
v^k_m(\a) = \left\{
\begin{array}{ll}
\arg\max_{n \in \mathcal{B}^k_m(\a)} H_{m, n}, &\text{if}~~\mathcal{B}_m^k(\a) \neq \emptyset\\
n_m, &\text{if}~~\mathcal{B}_m^k(\a) = \emptyset
\end{array}
\right.
\end{equation}
where $n_m$ is the home BS of UE $m$. Therefore, the amount of data $\lambda_{m,n}$ sent from UE $m$ to BS $n$ can be determined as
\begin{align}
\lambda_{m,n} = \sum_{k}\lambda^k_m \cdot\mathbf{1}\{v^k_m(\a)=n\}
\end{align}
and the amount of workload processed by BS $n$, $\gamma_{m,n}$, and cloud, $\gamma_{m,0}$, can be determined as:
\begin{align}
&\gamma_{m,n} = \sum_{k} \gamma^k_m \cdot\mathbf{1}\{v^k_m(\a)=n\}\cdot \mathbf{1}\{k \in a_n\}\\
&\gamma_{m,0} = \sum_{k} \gamma^k_m \cdot\mathbf{1}\{v^k_m(\a)=n_m\}\cdot \mathbf{1}\{k \not\in a_n\}
\end{align}
Substituting $\lambda_{m,n}$ into \eqref{TXcost} and $\gamma_{m,n}, \gamma_{m,0}$ into \eqref{PXcost}, we can then obtain the utility $U_n(\a)$ of SC $n$ which is a function of the service caching configuration $\a$ of all SCs.

The objective of cooperative SC network is to determine the optimal collaborative service caching configuration $\a$ in order to maximize the total utility of the whole network:
$\max\limits_{\a} \sum_n U_n(\a)$, which is equivalent to minimizing the total cost. Define the cost $C_n(\a)$ of SC $n$ is as
\begin{align}
C_n(\a) \triangleq  \sum_{m\in\mathcal{M}_n} \left( P_m \sum_{k}\lambda^k_{m,n}/r_{m,n} + \sum_{n\in\mathcal{N}} c_n \gamma_{m,n} + c_0\gamma_{m,0}\right)
\end{align}
Then the problem of collaborative service  caching for obedient SCs becomes:
\begin{align}\textbf{CSC-O:}~~
\min_{\a = (a_1, ..., a_N)}~~\sum_n C_n(\a)~~~\text{s.t.}~a_n \in\mathcal{F}_n, \forall n
\end{align}

The above problem is a difficult combinatorial optimization problem where the service caching decisions of SCs are highly correlated. Centralized solutions are usually computationally prohibitive and require the knowledge of the entire network, which is difficult to obtain. In the next subsection, we develop an efficient decentralized algorithm to optimize the collaborative service caching decisions based on parallel Gibbs sampler and graph coloring.

%\subsection{SC Grouping}
\subsection{Chromatic Parallel Gibbs Sampler and Graph Coloring}
Our decentralized collaborative service caching algorithm is developed based on the Gibbs sampling (GS) technique \cite{robert2004monte}. In our problem, GS is used to generate the probability distribution of network service caching configuration $\a$ by repeatedly sweeping each SC sampling its service caching decision from conditional distribution with the remaining SCs fixing their service caching decisions. The theory of Markov chain Monte Carlo (MCMC) guarantees that the probability distribution of service caching configuration $\a$ approximated by GS is proportional to $e^{-\sum_nC_n(\a)/\tau}$, where the parameter $\tau$ has the interpretation of temperature in the context of statistical physics. Such a probability measure is known as the Gibbs distribution. Furthermore, performing Gibbs sampling while reducing $\tau$ can obtain service caching configurations resulting in globally minimal cost \cite{geman1984stochastic}. However, convergence results are typically available for the case of sequential sampling (one SC updating at a time) only. The main problem of sequential sampling is that: (1) it takes too long to complete one round of updating for large networks, (2) it works with an additional assumption that the global communication is available, which may not hold in SC networks. One promising solution for this problem is to enable parallelism in Gibbs sampling. It is worth noting that extreme parallelism (i.e. all SCs evolve their decisions at the same time) is often infeasible. As \cite{newman2008distributed} have observed, one can easily construct cases where the parallel Gibbs sampler in which all SCs evolve their actions at the same time is not ergodic and therefore convergence to the Gibbs distribution is not guaranteed. In the following, we design an algorithm, called Chromatic Parallel Gibbs Sampler (CPGS), to transform the sequential sampling into an equivalent parallel sampling by exploiting the special structure of the considered SC network using the theory of Markov Random Field and Graph Coloring.

The sequential GS works by iteratively sampling each SC's service caching decisions according to a joint posterior distribution $p$, i.e.
\begin{align}\label{gibbs_sampler}
    a_i \sim p(a_i | \a_{-i}) \triangleq \frac{\exp(-\sum_n C_n(a_i, \a_{-i})/\tau)}{\sum_{a'\in\mathcal{F}_i}\exp(-\sum_n C_n(a', \a_{-i})/\tau)}, \forall i \in \mathcal{N}
\end{align}
where $\a_{-i}$ refers to the caching decisions of SCs excluding the SC $i$, and \textit{temperature} $\tau$ controls the trade-off between exploitation and exploration. The adopted distribution $p$ will ensure that GS converges to the Gibbs distribution. Intuitively, if the decision update at SC $i$ does not influence sampling distribution $p(a_j|\a_{-j})$ of SC $j$ and vice versa, then SC $i$ and $j$ can update their caching decisions independently and simultaneously. To formalize this intuition, we resort to the Markov Random Field (MRF). The MRF of a distribution $p$ is an undirected graph over the (caching decision variables of) SCs. On this graph, the set of all SCs adjacent to SC $i$, denoted by $\Gamma_i$, is called the Markov Blanket of SC $i$. Given its Markov Blanket, the caching decision of SC $i$ must be conditionally independent of the caching decisions of all other SCs outside $\Gamma_i$, namely
 \begin{align}
 	p(a_i | \a_{\Gamma_i})=p(a_i | \a_{-i})
 \end{align}
Therefore, any two SCs that are not in the Markov Blanket of each other can evolve their decisions simultaneously. The next proposition establishes the connection between the physical network graph and the MRF.

\begin{proposition}
The two-hop neighborhood of SC $i$ on the physical network graph $G$ is the Markov Blanket of SC $i$ on the MRF.
\end{proposition}
\begin{proof}
Consider any SC $j$ that is more than two hops away from SC $i$. We need to prove below:
\begin{align}
p(a_i|\a_{-i\backslash\{j\}}, a_j) = p(a_i|\a_{-i\backslash\{j\}}, a'_j), \forall a_j \neq a'_j
\end{align}
Recall that $\Omega_i$ is the one-hop neighborhood of SC $i$ on the physical network graph $G$. The total cost of the network can be divided into two parts
\begin{align}
\sum_n C_n(a_i, \a_{-i}) = \sum_{n\in\Omega_i}C_n(a_i, \a_{-i}) + \sum_{n\not\in\Omega_i}C_n(a_i, \a_{-i})
\end{align}
For any SC $n \in \Omega_i$, SC $j$ is not its one-hop neighbor because SC $j$ is at least three hops away from SC $i$. Therefore, fixing the caching decisions of all SCs except SC $i$ and SC $j$, $C_n(a_i, \a_{-i})$ only depends on the SC $i$'s decision, namely $C_n(a_i, \a_{-i}) = C_n(a_i, \a_{-i\backslash\{j\}})$.
On the other hand, for any SC $n \not\in \Omega_i$, SC $i$ is not its one-hop neighbor and hence, changing SC $i$'s decision does not affect $C_n(a_i, \a_{-i})$, namely $C_n(a_i, \a_{-i}) = C_n(\a_{-i})$. Therefore, 
\begin{align}
\frac{\exp(-\sum_n C_n(a_i, \a_{-i})/\tau)}{\sum_{a'\in\mathcal{F}_i}\exp(-\sum_n C_n(a', \a_{-i})/\tau)} = \frac{\exp(-\sum_{n\in\Omega_i} C_n(a_i, \a_{-i\backslash\{j\}})/\tau)}{\sum_{a'\in\mathcal{F}_i}\exp(-\sum_{n\in\Omega_i} C_n(a', \a_{-i\backslash\{j\}})/\tau)}
\end{align}
The right-hand side is independent of $a_j$, the proposition is proved. It is also obvious that the one-hop neighborhood $\Omega_i$ of SC $i$ on the physical network graph is not its Markov Blanket.
\end{proof}

Proposition 1 implies that when SC $i$ changes its caching decision, the change in the sum cost $\sum_{n \in \Omega_i} C_n(\a)$ of SC $i$'s one-hop neighborhood is the same no matter how a SC $j$ outside of its Markov Blanket also changes its caching decision at the same time. This property is the key to enabling the parallel evolution of the service caching decisions in GS. Based on this result, we divide the set of SCs into $L \leq N$ groups such that no two SCs within a group are in each other's Markov Blanket and hence, SCs within the same group can evolve their decisions in parallel. We would like $L$ to be as small as possible to achieve the maximal level of parallelization.

%Therefore, all SCs in the same time can search for their optimal service caching decisions in parallel.

%Based on this result, we color the SC network with $l$-coloring such that each SC is assigned one of $l$ colors and correlated SCs have different colors. Let $L_i$ denote the SCs in color $i$. Then the chormatic sampler simutaneously draws new decisions for all SCS in $l_i$ Therefore, all SCs in the same time can search for their optimal service caching decisions in parallel.

Finding the minimum value of $L$ is equivalent to a graph coloring problem on the MRF. We color the MRF network with $L$-coloring such that each SC is assigned one of $L$ colors and SCs in the Markov Blanket have different colors. Therefore, the minimum value of $L$ equals the chromatic number of the graph MRF. Let $l_i$ denote the set of SCs in color $i$. Then CPGS simultaneously draws new decisions for all SCs in $l_i$ before proceeding to $l_{i+1}$. The colored network ensures that all SCs within a color are conditionally independent of the SCs in the remaining colors and therefore can be sampled in parallel. Constructing the minimal coloring of a general network is NP-Complete. However, for simple models the optimal coloring can be quickly derived using existing graph coloring algorithms \cite{kubale2004graph}.

\textbf{Example}: Figure \ref{MRF} illustrates the relationship between the physical network graph, the MRF, the Markov blanket and the coloring result for two typical network topologies, namely circle and star. For the circle network with $N = 6$, the minimum number of SC groups is $L = 3$. For the star network with $N = 9$, the minimum number of SC groups is $L = 5$. In both cases, $L$ is significantly smaller than $N$.

\begin{figure}[htb]
  \centering
  \includegraphics[width=0.7\textwidth]{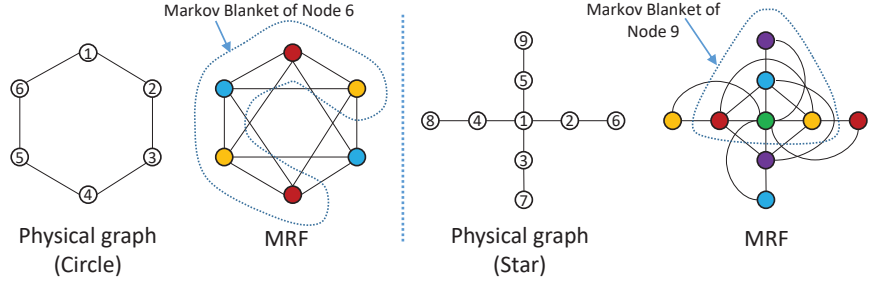}\\
  \vspace{-0.15 in}
  \caption{Illustration of MRF, Markov Blanket, and coloring.}\label{MRF}
  \vspace{-0.18 in}
\end{figure}

\subsection{Decentralized Algorithm Based On Chromatic Parallel Gibbs Sampling}
Now we are ready to present the decentralized algorithm to optimize the collaborative caching decisions (the pseudo-code is provided in Algorithm \ref{dis_alg_p2}). The algorithm works distributedly in an iterative manner as illustrated in Fig. \ref{alg_illu}. In each iteration $t$, a colorset $l_j$ is chosen according to a prescribed order. Every SC $i \in l_j$ in this colorset goes through two steps: decision update and communication. To perform decision update, SC $i$ needs two pieces of information: (1) the service demand patten of its one-hop neighbors, which is exchanged at the beginning of each caching decision period; (2) the current service caching decision of SCs within its \textit{Markov Blanket} $\Gamma_i$, which are collected in the previous iterations. With this information, SC $i$ is able to compute the total cost of its one-hop neighborhood $\sum_{n \in \Omega_i}C_n(a_i, \a_{\Gamma_i}(t))$ locally for every possible caching decision $a_i \in \mathcal{F}_i$ while fixing the caching decisions of the SCs in $\Gamma_i$. Then, the BS samples a new service caching $a_i(t+1) = a_i$ based on the following probability distribution:
\begin{align}\label{samp_dist}
p\left(a_i \mid \a_{\Gamma_i}(t) \right) = \frac{\exp(-\sum_{n \in \Omega_i}C_n(a_i, \a_{\Gamma_i}(t))/\tau)}{\sum_{a'_i\in\mathcal{F}_i}\exp(-\sum_{n \in\Omega_i} C_n(a'_i, \a_{\Gamma_i}(t))/\tau)}, \forall a_i \in \mathcal{F}_i
\end{align}
The above equation implies that an action is selected with a higher probability if it leads to a lower total cost of the one-hop neighborhood. After the service caching decision $a_i$ is updated, the chosen SCs communicate new service caching decisions to the SCs in $\Gamma_i$, which prepares for the subsequent iterations. Notice that during the iterations, SCs do not need to actually change their service caching decision, which is only needed after the completion of the algorithm. 

Next, we formally prove the convergence and optimality of our algorithm.
\begin{figure}[htb]
	\centering
	\includegraphics[width=0.95\linewidth]{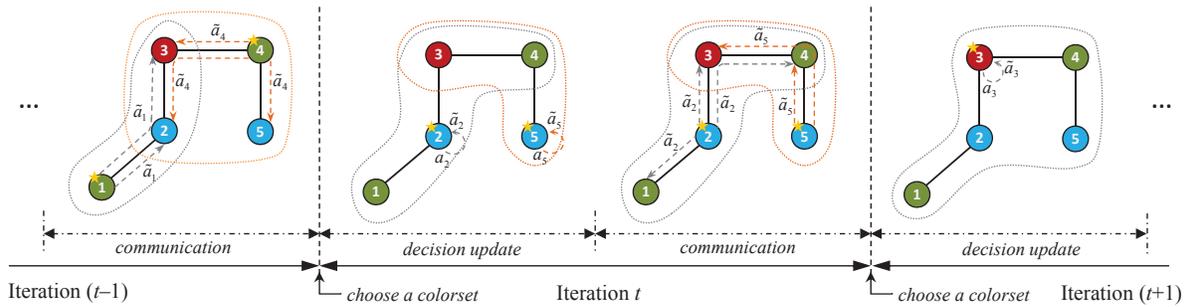}
	\vspace{-0.15 in}
	\caption{Illustration of decentralized collaborative service caching algorithm. The figure depicts a physical SC network; the \textit{star sign} denotes the chosen SCs in each iteration to update their service caching decision.}
	\label{alg_illu}
\end{figure}

\begin{algorithm}[htb]
\caption{Decentralized Collaborative Service Caching}
\begin{algorithmic}[1]
\STATE \textbf{Input}: Physical network graph $G$, UE demand $\d$.
\STATE Construct correlation graph $G^c$ based on $G$;
\STATE Determine $L$ groups of uncorrelated SCs using graph coloring algorithm;
\FOR {each iteration $t$}
\STATE Pick a group of uncorrelated SCs $l_j$ according to some prescribed order;
\FOR {each SC $i \in l_j$ in the picked group}
\FOR {each feasible caching decision $a_i \in \mathcal{F}_i$}
\STATE Determine the workload distribution of UEs in $\Gamma_i$;
\STATE Compute the total cost $\sum_{n \in \Omega_i}C_n(a_i, \a_{\Gamma_i}(t))$ of the one-hop neighborhood $\Omega_i$;
\ENDFOR
\STATE Set $a_i(t+1) \gets a_i$ with probability \eqref{samp_dist} and communicate $a_i(t+1)$ to SCs in $\Gamma_i$;
\ENDFOR
\FOR {each SC $i \not \in l_i$ not in the picked group}
\STATE Set $a_i(t+1) \gets a_i(t)$;
\ENDFOR
\STATE Stop iteration if the sampling probability converges; %\textit{Stopping Criteria}
\ENDFOR
\STATE \textbf{Return}: Collaborative service caching configuration $\a$.
\end{algorithmic}\label{dis_alg_p2}
\end{algorithm}

\vspace{-0.1 in}
\begin{proposition} [Convergence] \label {theo_converge}
The Chromatic Parallel Gibbs Sampler converges from any starting configuration to Gibbs distribution $\pi(\a)$, where
\begin{align}
\pi(\a)=\dfrac{e^{-\sum_{n\in\mathcal{N}}C_n(\a)/\tau}}{\sum_{\a\in\mathcal{F}_1 \times \dots \times \mathcal{F}_N} e^{-\sum_{n\in\mathcal{N}}C_n(\a)/\tau}}
\end{align}
\end{proposition}

\begin{proof}
	The convergence theorem of CPGS can be easily derived from the classic convergence of sequential Gibbs sampler \cite{geman1984stochastic} presented below:
	\begin{lemma}[Convergence of sequential Gibbs sampler] \label{conv_seq_Gib}
		Assume that each SC $n\in\mathcal{N}$ is sampled by the sequential Gibbs sampler with the sequence $\{n_t, t\geq 1\}$ containing SC $n$ infinitely often. Then, for every starting configuration $\a(0)$ and every $\a$, we have $\lim_{t\to\infty} P\left(\a(t)=\a \mid \a(0)\right)=\pi(\a)$ where $\a(0), \a \in \mathcal{F}_1 \times \dots \times \mathcal{F}_N$.
	\end{lemma}
   \begin{proof}
    	The proof of this Lemma \ref{conv_seq_Gib} can be found in \cite{geman1984stochastic}, thus is omitted here.
   \end{proof}

    For sequential Gibbs sampling, we can construct a sampling sequence as:
    \begin{align}
     \{n_t, t \geq 1\}:~\underbrace{l_{1,1},l_{1,2},\dots}_{|l_1|},\underbrace{l_{2,1},l_{2,2},\dots}_{|l_2|},\dots,\underbrace{l_{L,1},l_{L,2},\dots}_{|l_L|},\dots
    \end{align}
    where $l_{j,i}$ denotes $i$-th BS in colorset $l_j$. Since the convergence of sequential Gibbs sampling is guaranteed given any prescribed order as long as every SC is visited infinitely often, we know that sampling with the constructed order converges to the Gibbs distribution. Recall that in CPGS the caching decision updates are independent for SCs sharing the same color, therefore, SCs with the same color can update the caching decision at the same time. This also ensures the convergence to the Gibbs distribution.
\end{proof}

By extending Proposition \ref{theo_converge}, we can show the algorithm also guarantees that the service caching configuration $\a$ converges to the optimal solution of \textbf{CSC-O} with an appropriate temperature $\tau$.

\begin{proposition} [Optimality] \label {theo_optim}
	As $\tau$ decreases, the algorithm converges with a higher probability to the optimal solution of \textbf{CSC-O}. When $\tau \to 0$, the algorithm converges to the global optimal solution with probability 1.
\end{proposition}

\begin{proof}
	Let $\tilde{C}^{opt} = \min_\a \sum_{n\in\mathcal{N}} C_n(\a)$, we rewrite $\pi(\a)$ as
    \begin{align}
       \hat{\pi}(\a)=\dfrac{e^{-\left(\sum_{n\in\mathcal{N}}C_n(\a)-\tilde{C}^{opt}\right)/\tau}}{\sum_{\a\in\mathcal{F}_1 \times \dots \times \mathcal{F}_N}e^{-\left(\sum_{n\in\mathcal{N}}C_n(\a)-\tilde{C}^{opt}\right)/\tau}}
    \end{align}
 where $\hat{\pi}(\a)$ is obtained from $\pi(\a)$ by multiplying both its denominator and numerator by $e^{\tilde{C}^{opt}/\tau}$. As $\tau \to 0$, $e^{-\left(\sum_{n\in\mathcal{N}}C_n(\a)-\tilde{C}^{opt}\right)/\tau}$ approaches 1 if $\sum_{n\in\mathcal{N}}C_n(\a)=\tilde{C}^{opt}$ and approaches 0 otherwise. As a result, the optimal caching decision $\a^{opt} = \arg\min _{\a} \sum_{n\in\mathcal{N}} C_n(\a)$ will be sampled with probability 1.
\end{proof}
%\begin{corollary}\textbf{(Ergodic theorem)}
%	For the defined cost $\tilde{C}(\a)=\sum_{n\in\mathcal{N}}C_n(\a)$, its expectation exists:
%	\begin{align}
%		\lim_{T\to\infty}\dfrac{1}{T}\sum_{t=1}^{T} \tilde{C}(\a^t) \xrightarrow{\text{a.s.}} E_\pi [\tilde{C}(\a)]
%	\end{align}
%\end{corollary}
Following the classic result of parallel computing in \cite{bertsekas1989parallel}, we can analyze the time complexity of proposed algorithm:
\begin{proposition} [Time complexity]\label{prop_runtime}
	Given a $L$-coloring of the SC network, CPGS generates a new joint service caching decision in runtime $O(L)$.
\end{proposition}
\vspace{-0.15 in}
\begin{proof}
	From \cite{bertsekas1989parallel} (Proposition 2.6), we know that the parallel execution of CPGS corresponds exactly to the execution of a sequential scan Gibbs sampler for some permutation over variables. Therefore the running time can be derived as:
    \vspace{-0.1 in}
	\begin{align}
	O\left(\sum_{i=1}^L \left\lceil\dfrac{|l_i|}{\rho}\right\rceil\right) \stackrel{(\dag)}= O\left(\sum_{i=1}^L \dfrac{|l_i|}{|l_i|}\right) = O(L)
	\end{align}
	where $\rho$ is the number of processors. In our case, each SC derives caching decision individually at local edge server, therefore $\rho$ equals the number of SCs in $l_i$, which gives the equality $(\dag)$.
\end{proof}

Proposition \ref{prop_runtime} indicates that CPGS provides a linear speed-up for single-chain sampling, advancing the Markov chain for an $L$-coloring SC network in time $O(L)$ rather than $O(N)$. Typically, $L$ is much smaller than $N$, thereby accelerating the convergence speed of our algorithm.

\section{Collaborative Service Caching for Strategic SCs}
In the previous section, we studied collaborative service caching assuming that SCs are fully cooperative. However selfish SCs will be reluctant to participate in the cooperative network if doing so reduces their individual utility. To formally understand SCs' selfish behavior and the resulting service caching strategies, we use the theoretical framework of the coalitional game to investigate how selfish SCs form coalitions to collaboratively provide edge computing services.

A coalitional game is defined as a tuple $(\mathcal{N}, v)$ where $\mathcal{N}$ is the set of players and $v: 2^\mathcal{N} \to \mathbb{R}$ is a function that assigns for every possible coalition $\mathcal{S} \subseteq \mathcal{N}$ a real number representing the total benefit achieved by coalition $\mathcal{S}$. By evaluating the values of different coalitions, players decide what coalitions to form. In what follows, we first design the interaction within each coalition and define the value function $v(\mathcal{S})$ for any given coalition $\mathcal{S} \subseteq \mathcal{N}$. Then we describe what coalitions are desired in terms of stability and design a distributed coalition formation algorithm.

\vspace{-0.1 in}
\subsection{Value Function and SC Interactions within a Coalition}
In this subsection, we define the value function and describe the interaction among SCs that belong to a given coalition $\mathcal{S}$ (which, however, may not be stable). The value function $v(\mathcal{S})$ of an SC coalition $\mathcal{S}$ is defined as the maximum total utility that can be achieved when collaboration is restricted to SCs in this coalition, i.e.
\begin{align}
v(\mathcal{S}) = \max_{\a_{\mathcal{S}}}\sum_{n\in\mathcal{S}} U_n(\a_{\mathcal{S}})
\end{align}
where $\a_\mathcal{S}$ denotes the collaborative service caching decision of SCs in coalition $\mathcal{S}$. When cooperation is restricted in $\mathcal{S}$, an SC in $\mathcal{S}$ does not offload its UE workload to any SC outside $\mathcal{S}$. By using the collaborative service caching algorithm developed in Section \ref{sec_CSC} (restricted to coalition $\mathcal{S}$), the optimal total utility and hence the value of the coalition can be computed. However, $v(\mathcal{S})$ is achieved assuming that SCs are always cooperative within the coalition, which may not be true. It is entirely possible that some SCs cache popular services for neighbor SCs instead of its favorite services in order to improve utility of the coalition, which makes them unwilling to work in a collaborative manner. In the following, we design an effective algorithm for strategic SCs based on coalitional game such that collaboration is always favorable for SCs to work in the generated coalitions. Before proceeding, we introduce two collaboration patterns:
\subsubsection{Plain Collaboration}
The SC collaboration simply relies on \textit{pay-to-work} framework described in Section \ref{comcost}): each UE pays the BSs and cloud to cover the computation cost incurred by processing its service requests. This payment scheme is relatively weak since UEs only covers the computation cost of SC for processing the service requests and no extra reward for participating in the collaboration.

\subsubsection{Incentivized Collaboration}
Beside the \textit{pay-to-work} framework, SCs receive extra payments by participating in the collaboration. In this case, the values of extra payments are decided by \textit{Proportional Fairness Division} \cite{chen2017socially}: dividing the payoff (i.e. utility improvement) of the whole coalition due to cooperation proportionally to the SC's individual utility achieved without cooperation. Specifically, for SC $n$, its modified utility $\tilde{U}_n$ will be
\begin{align}
\tilde{U}_n = \psi_n\cdot\left(v(\mathcal{S}) - \sum_{i \in \mathcal{S}} v(\{i\})\right) + v(\{n\})
\end{align}

where $v(\{n\})$ represents the utility of SC $n$ if it does not join any coalition (perform non-collaborative service caching), and $\psi_n$ is the proportional weight satisfying $\sum_{n \in \mathcal{S}} \psi_n= 1$, i.e.,
\begin{align}
\psi_n/\psi_j = v(\{n\})/v(\{j\}),\forall n,j\in\mathcal{S}
\end{align}

Based on the proportional fairness division, the extra payment paid/received by SC $n$ can thus be determined as $y_n = \tilde{U}_n(\a_{\mathcal{S}}) - U_n(\a_{\mathcal{S}})$. Here, $\tilde{U}_n$ can be interpreted as the expected utility of SC $n$ by participating in the coalition, while $U_n$ is the actually realized utility. The gap of the two is filled by the extra payment $y_n$. Clearly, the $y_n$ must be cleared within each coalition and hence $\sum_{n \in \mathcal{S}}y_n = 0$. If $y_n > 0$, then SC $n$ pays $y_n$. If $y_n < 0$, then SC $n$ receives $|y_n|$. 

Intuitively, we have extra payment $y_n=0, \forall n$ (i.e., $\tilde{U}_n(\a_{\mathcal{S}}) = U_n(\a_{\mathcal{S}}$)) for SCs with the \textit{Plain Collaboration}. It is expected that \textit{Incentivized collaboration} can better promote cooperation among strategic SCs as will be shown later in the simulation. 

There are two implementation issues for the payment scheme. First, payments need to be properly distributed since multiple SCs may be involved in the transaction. Moreover, direct payment from one SC to another faces fraud risks in the monetary transaction. To enable effective and safe transaction, the edge orchestrator, which is a trusted third-party, collects payments from all SCs and then distributes them to SCs.

\vspace{-0.1 in}
\subsection{Stability of Coalitions}
SCs may form multiple disjoint coalitions and there are many ways that SCs can form coalitions. To characterize what kind of coalitions are preferred by SCs, we introduce the notion of \textit{Stable Coalition}: no SC(s) have incentives to leave the current coalition to form different coalitions. Clearly, the requirement that all SCs in a coalition must receive higher utilities than working individually is a necessary (but not sufficient) condition for stability.

Consider any subset $\tilde{\mathcal{N}} \subseteq \mathcal{N}$ of SCs. We call $\mathcal{P} = \{\mathcal{S}_1,\mathcal{S}_2,..., \mathcal{S}_W\}$ a collection of coalitions formed by SCs in $\tilde{\mathcal{N}}$, where $\mathcal{S}_w \subseteq \tilde{\mathcal{N}}, \forall w$ are disjoint subsets of $\tilde{\mathcal{N}}$. If $\tilde{\mathcal{N}} = \mathcal{N}$, then we call $\mathcal{P}$ a partition of $\mathcal{N}$. We introduce the notion of a defection function $\mathbb{D}$, which associates each possible partition $\mathcal{P}$ of $\mathcal{N}$ with a group of collections. The stability of a partition of $\mathcal{P}$ is defined against a defection function as follows.

\begin{definition}[$\mathbb{D}$-stability]
	A partition $\mathcal{P}$ of $\mathcal{N}$ is $\mathbb{D}$-stable if no SCs are interested in leaving the current partition $\mathcal{P}$ to form a new collection of coalitions $\mathcal{P}' \in \mathbb{D}(\mathcal{P})$. That is, at least one SC in these SCs does not improve its utility by leaving the current partition.
\end{definition}

In other words, a defection function $\mathbb{D}$ restricts the possible ways that SCs may deviate/defect. Two defection functions are of particular interest. The first function, denoted by $\mathbb{D}_c$, associates each partition $\mathcal{P}$ with all possible collections in $\mathcal{N}$. Therefore, it does not put any restriction on the way SCs may deviate and hence is the most general case. The second function, denoted by $\mathbb{D}_{hp}$, associates each partition $\mathcal{P}$ with collections that can be formed by merging or splitting coalitions in $\mathcal{P}$. Therefore, $\mathbb{D}_{hp}$-stability is a weaker notion than $\mathbb{D}_c$-stability. Next, we design a distributed SC coalition formation algorithm that achieves at least $\mathbb{D}_{hp}$-stability.
\vspace{-0.1 in}
\subsection{Distributed Coalition Formation Algorithm for Small-cell network}
Before presenting the SC coalition formation algorithm, we first introduce the notion of \textit{Pareto dominance} to compare the ``quality'' of two collections of coalitions.
\begin{definition}[Pareto Dominance]
	Consider two collections of disjoint coalitions $\mathcal{S}_1$ and $\mathcal{S}_2$ formed by the same subset of SCs $\tilde{\mathcal{N}} \subseteq \mathcal{N}$. $\mathcal{S}_1$ Pareto-dominates $\mathcal{S}_2$, denoted by $\mathcal{S}_1\rhd\mathcal{S}_2$, if and only if $\tilde{U}_i(\a_{\mathcal{S}_1}) \geq \tilde{U}_i(\a_{\mathcal{S}_2}), \forall i \in \tilde{\mathcal{N}}$ with at least one strict inequality for some SC.
\end{definition}

Pareto dominance implies that a group of SCs prefer to form coalitions $\mathcal{S}_1$ rather than $\mathcal{S}_2$. Based on this, we define following two operations, namely \emph{Merge} and \emph{Split} \cite{apt2009generic}, which are central to our coalition formation algorithm:
\begin{itemize}
	\item \textbf{Merge}: merge a set of coalitions $\{\mathcal{S}_1,\dots,\mathcal{S}_l\}$ into a bigger coalition $\bigcup^{l}_{j=1}\mathcal{S}_j$ if  $\{\bigcup^{l}_{j=1}\mathcal{S}_j\}\rhd\{\mathcal{S}_1,\dots,\mathcal{S}_l\}$.
	\item \textbf{Split}: split a coalition $\{\bigcup^{l}_{j=1}\mathcal{S}_j\}$ into a set of smaller coalitions $\{\mathcal{S}_1, \dots, \mathcal{S}_l \}$ if $\{\mathcal{S}_1, \dots, \mathcal{S}_l \}\rhd\{\bigcup^{l}_{j=1}\mathcal{S}_j\}$.
\end{itemize}

By performing Merge, a group of SCs can operate and form a single and larger coalition if and only if this formation increases the utility of at least one SC without decreasing the utility of any other involved SCs. Hence, a Merge decision ensures that all involved SCs agree on its occurrence. Likewise, a coalition can decide to split and divide itself into smaller coalitions if splitting is preferred in the Pareto sense.

\begin{algorithm}[htb]
	\caption{Distributed SC coalition formation}
	\begin{algorithmic}[1]
		\STATE \textbf{Initial}: The SC network is partitioned by $\mathcal{P}=\left\{\{1\},\{2\}\dots,\{N\}\right\}$ with non-collaborative SCs at the beginning of each operational time slot.
		\STATE \textbf{Output}: A stable partition $\mathcal{P}_f=\{\mathcal{S}_1,\mathcal{S}_2,\dots,\mathcal{S}_W\}$ of $\mathcal{N}$; Service decision of SCs in each coalition $\mathcal{S}_w\in \mathcal{P}_f$: $\a=\{\a_{\mathcal{S}_1},\dots,\a_{\mathcal{S}_W}\}$.
		
		%\STATE \textit{\textbf{Phase 1: SC Coalition Formation}}
		\STATE \textbf{Repeat}
	    
		\STATE (a) $\mathcal{P}^{\prime}\leftarrow$ Merge ($\mathcal{P}$):
		    \begin{ALC@g}
		    \STATE Chose a possible set of SC coalitions $\{\mathcal{S}_1,\dots,\mathcal{S}_l\} \subseteq \mathcal{P}$ for attempt of merging into $\mathcal{S}^\prime=\bigcup^{l}_{j=1}\mathcal{S}_j$, where $\mathcal{S}^\prime\in\mathcal{P}^\prime$;
	     	\STATE Solve $v(\mathcal{S}^\prime) = \max\limits_{\a_{\mathcal{S}^\prime}}\sum\limits_{n\in\mathcal{S}^\prime} U_n(\a_{\mathcal{S}^\prime})$ and obtain optimal configuration $\a^{opt}_{\mathcal{S^\prime}}$ with CPGS;
	     	\STATE Derive $\tilde{U}_i, \forall i\in \mathcal{S}^\prime$, and decide whether to merge by examining Pareto dominance;
	     	\end{ALC@g}
		\STATE (b) $\mathcal{P}\leftarrow$ Split ($\mathcal{P}^{\prime}$)
	        \begin{ALC@g}
			\STATE Chose a coalition $\mathcal{S}^\prime\in \mathcal{P^\prime}$ for attempt of splitting into a set of disjoint coalitions  $\mathcal{S}^\prime\to \{\mathcal{S}_1,\dots,\mathcal{S}_l\}$, where $\mathcal{S}_1,\dots,\mathcal{S}_l \in \mathcal{P}$;
			\STATE Solve $v(\mathcal{S}_l) = \max\limits_{\a_{\mathcal{S}_l}}\sum\limits_{n\in\mathcal{S}_l} U_n(\a_{\mathcal{S}_l}), \forall l$ and obtain optimal configuration $\a^{opt}_{\mathcal{S}_l}$ with CPGS;
			\STATE Derive $\tilde{U}_i, \forall i\in \{\mathcal{S}_1,...,\mathcal{S}_l\} $, and decide whether to split by examining Pareto dominance;
	    	\end{ALC@g}
      
		\STATE \textbf{Until} Merge and Split converges to a stable partition $\mathcal{P}_f$;
		\STATE \textbf{Return}  Stable partition $\mathcal{P}_f$ and its service caching decisions $\a=\{\a_{\mathcal{S}_1},\dots,\a_{\mathcal{S}_W}\}$.
     	
%		\STATE \emph{Phase 2: Cooperative computation offloading}
%		\STATE \quad (a) each coalition $S_i\in\mathcal{S}_f$ order its buyers in a way to minimize the offloading cost (maximize \eqref{value_function}).
%		\STATE \textbf{Repeat} for every $S_i\in \mathcal{S}_f$
%		\STATE \quad (b) each buyer in a coalition $S_i\in\mathcal{S}_f$ attempts to acquire computation demands in coalition $S_i$.
%		\STATE \textbf{Until} no peer offloading in the coalition is available.
%		\STATE \quad (c) any buyer who still has unsatisfied computation demand performs SC-to-cloud offloading.
%		\STATE These two stages are repeated periodically to adapt the partition to environmental changes.
		
	\end{algorithmic}\label{alg_coalition_form}
\end{algorithm}

Our coalition formation algorithm for collaborative service caching with strategic SCs (\textbf{CSC-S}) is developed based on the Merge and Split operations, which is presented in Algorithm \ref{alg_coalition_form}. The algorithm iteratively executes the \textit{merge-and-split} operations. Given the current partition $\mathcal{P}$, each coalition $\mathcal{S} \in \mathcal{P}$ negotiates, in a pairwise manner, with neighboring SCs to assess a potential merge. The two coalitions will then decide whether or not to merge. Whenever a Merge decision occurs, a coalition can subsequently investigate the possibility of a Split. Clearly, a Merge or a Split operation is a distributed decision that an SC (or a coalition of SCs) can make. After successive \textit{merge-and-split} iterations, the network converges to a partition composed of disjoint coalitions and no coalition has any incentive to further merge or split. In other words, the partition is \textit{merge-and-split} proof. The convergence of any \textit{merge-and-split} iterations such as the proposed algorithm is guaranteed as shown in \cite{apt2009generic}.

\subsection{Stability Analysis}
The outcome of the above algorithm is a partition of disjoint independent coalitions of SCs. As an immediate result of the definition of $\mathbb{D}_{hp}$ stability, every partition resulting from proposed algorithm is $\mathbb{D}_{hp}$-stable. In particular, no coalitions of SCs in the final partition have the incentive to pursue a different coalition formation through Merge or Split. Next, we investigate whether the proposed algorithm can achieve $\mathbb{D}_{c}$-stability.

A $\mathbb{D}_c$-stable partition has the following properties according to \cite{apt2009generic}. (i) No SCs are interested in leaving $\mathcal{S}$ to form other collections in $\mathcal{N}$ (through any operation). (ii) A $\mathbb{D}_c$-stable partition is the \textit{unique} outcome of any \textit{arbitrary} iteration of merge-and-split if it exists. (iii) A $\mathbb{D}_c$-stable partition $\mathcal{P}$ is a unique $\rhd$-maximal partition, i.e., for all partition $\mathcal{P}^\prime\neq\mathcal{P}$, we have $\mathcal{P}\rhd\mathcal{P}^\prime$. Therefore, the $\mathbb{D}_c$-stable partition provides a \textit{Pareto optimal} utility distribution. However, the existence of a $\mathbb{D}_c$ stable partition is not always guaranteed \cite{apt2009generic}. Nevertheless, we can still have the following result.
\begin{proposition} [Coalition Stability]
	The proposed distributed SC coalition formation algorithm converges to the Pareto-optimal $\mathbb{D}_c$-stable partition, if such a partition exists. Otherwise, the final partition is $\mathbb{D}_{hp}$-stable.
\end{proposition}
\begin{proof}
	The proof is immediate due to the fact that, when it exists, the $\mathbb{D}_c$-stable partition is a unique outcome of any \textit{merge-and-split} iteration \cite{apt2009generic}, such as any partition resulting from our coalition formation algorithm.
\end{proof}

The stability of the grand coalition (e.g. all SCs form a single coalition) is of particular interest in the coalitional game theory. It can be easily shown that the considered SC coalitional game is generally not superadditive and its core is generally empty due to the limitation on BS interaction and hence, the grand coalition is not stable. Instead, independent disjoint coalitions will form. Readers who are interested in more details on the stability of grand coalition in coalitional games are referred to \cite{saad2009coalitional,bogomolnaia2002stability}.

\vspace{-0.1 in}
\subsection{Complexity Analysis}
The complexity of the proposed coalition formation algorithm lies mainly in the complexity of the Merge and Split operations. For a given network, in one Merge operation, each current coalition attempts to merge with other coalitions in a pairwise manner. In the worst case scenario, every SC, before finding a suitable merge partner, needs to make a merge attempt with all other SCs in $\mathcal{N}$. In this case, the first SC requires $N-1$ attempts for Merge, the second requires $N-2$ attempts and so on. The total number of Merge attempts in the worst case is thus $\sum^{N}_{i=1}(N-i)$. In practice, the merge process requires a significantly lower number of attempts since finding a suitable partner does not always require to go through all possible merge attempts (once a suitable partner is identified the merge will occur immediately). The complexity is further reduced due to the fact that SCs do not need to attempt to Merge with physically unreachable SCs. Moreover, after the first run of the algorithm, the initial $N$ non-collaborative SCs will self-organize into larger coalitions. Subsequent runs of the algorithm will deal with a network composed of a number of coalitions that is much smaller than $N$.

For the split operation, in the worst case scenario, splitting can involve finding all the possible partitions of the set formed by the SCs in a single coalition. For a given coalition $\mathcal{S}$, this number is given by the Bell number $\sum^{|\mathcal{S}|}_{k=1} \binom{|\mathcal{S}|}{k}$ which grows exponentially with the number of SCs $|\mathcal{S}|$ in the coalition. In practice, this Split operation is restricted to the formed coalitions, and thus it will be applied to small sets. The split complexity is further reduced due to the fact that, in most scenarios, a coalition does not need to search for all possible split forms. For instance, once a coalition identifies a suitable split structure, the SCs in this coalition will split, and the search for further split forms is not needed in the current iteration.

\section{Simulation}
In this section, we conduct systematic simulations in practical scenarios to verify the efficacy of the proposed algorithms.

\vspace{-0.1 in}
\subsection{Setup}
Our system model adopts the widely-used stochastic geometry approach, homogeneous Poisson Point Process (PPP) \cite{baccelli2010stochastic}, for BS and UE deployment. To be specific, we simulate a 500m$\times$500m square area where a set of BSs are deployed whose locations are chosen according to the PPP with density $\lambda_{\text{BS}}=0.03$. The deployment of UEs follows another PPP with density $\lambda_{\text{UE}}=0.12$. The service demand of each UE is modeled as a Poisson process with arrival rate randomly drawn from 0 to 20. The transmission power of UEs is set as 10 dBm and the channel model follows the free space pathloss model: $20\log_{10}D_{m,n}+32.44$, where $D_{m,n}$  is the distance between UE $m$ and BS $n$. The bandwidth for each UE is $W=20$MHz. The cache capacity of BSs is set to 1 (service) and BSs are allowed to collaborate with other BSs with a range of 150m. The processing cost of BSs are random variables $c_n\in[1, 4]$ and the cost of cloud usage $c_0$ is set as 5. The temperature $\tau$ for chromatic parallel Gibbs sampler is set as 10.

Fig. \ref{deployment} shows the deployed BSs and UEs generated by homogeneous PPP. We have a total of 13 BSs and 72 UEs, and each UE registers to the BS nearest to it. Fig. \ref{demand} depicts the service demand received by BSs from its registered UEs, where the length of each color block denotes the demand for a particular service. The simulation results to be shown in the following are based on this BS/UE deployment and service demand pattern.

\begin{figure*}[htb]
	\begin{minipage}[t]{0.45\linewidth}
		\includegraphics[width=0.85\textwidth]{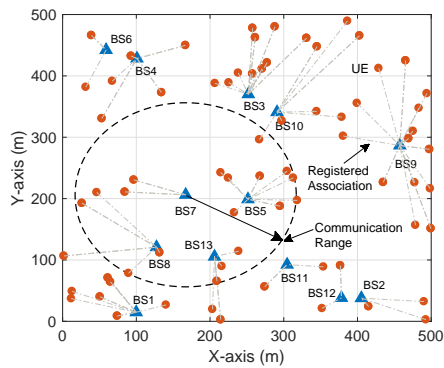}
		\vspace{-0.15 in}
		\caption{BS/UE deployment and registered association}
		\label{deployment}
	\end{minipage}%
	\begin{minipage}[t]{0.55\linewidth}
		\includegraphics[width=0.9\textwidth]{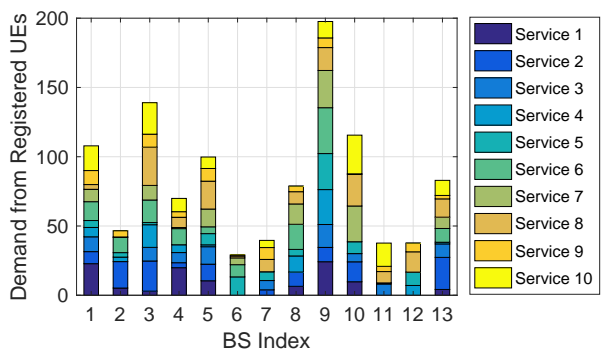}
		\vspace{-0.15 in}
		\caption{Demands of BSs received from registered UEs}
		\label{demand}
	\end{minipage}%
\end{figure*}

\subsection{Utility analysis}
First, we compare the system utilities achieved by three strategies: Non-collaborative service caching (NCOL), Collaborative service caching with obedient SCs (CSC-O) and Collaborative service caching with strategic SCs (CSC-S). For CSC-S, we applied both \textit{Plain Collaboration} (PC) and \textit{Incentivized Collaboration} (IC) schemes to see how they may influence the coalition formation. Fig. \ref{overall_utility} depicts the overall utility of the edge system by applying these four strategies. It can be observed that collaborative strategies dramatically increase the system utility compared to the non-collaborative case, achieving utility improvement by 42.8\% to 57.1\%. Specifically, CSC-O achieves the largest system utility since service caching decisions of all BSs are made to maximize the utility of the whole edge system. However, system-wide optimality may not be preferred for every BS, since not all BSs are guaranteed with utility gain by following the system optimal configuration. To see this, we show the utility achieved by individual BSs in Fig.\ref{AP_utility}. As can be observed, BS 6 and 11 suffer from utility loss compared with NCOL by following the optimal service caching decision derived by CSC-O. By contrast, the utilities of individual BSs achieved by CSC-S are strictly larger than that of NCOL case, which means the collaboration between BSs is favorable even considering the selfishness of SCs. Moreover, the system utility achieved by CSC-S is comparable with CSC-O.
\begin{figure}[htb]
	\centering	
	\includegraphics[width=0.55\textwidth]{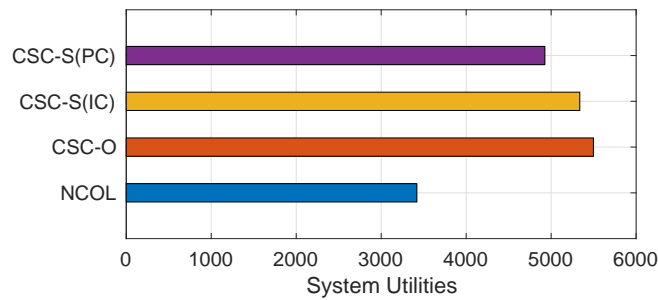}
	\vspace{-0.15 in}
	\caption{Comparison of system utilities}
	\label{overall_utility}
	\vspace{-0.25in}
\end{figure}

\subsection{Workload distribution and service caching decision}
Next, we proceed to see how the proposed algorithm works to benefit the edge computing system. The key idea of SC collaboration is to accommodate more workload at the edge system thereby avoiding the high cost of cloud usage. Fig. \ref{work_distribution} presents the workload distribution of each BS. We see clearly that the three collaborative strategies, compared to the NCOL, allow more workload to be processed within the edge system instead of offloading to the cloud. The main mechanism underlying this advantage is increasing the diversity of cached services among the neighbor BSs. Fig. \ref{serv_caching_decision} depicts the service caching decisions obtained by NCOL, CSC-O, and CSC-S(IC). Notice that, for better illustration, we place the service caching decisions of neighbor BSs adjacent to each other in the figure. By combining the service demand pattern in Fig. \ref{demand}, we can see that NCOL simply caches the services having the highest demand and therefore, the case illustrated in Fig. \ref{noncop_cop} can easily occur. For example, two neighbor BS 5 and BS 7 both cache service 8 in NCOL case. By contrast, the service caching decisions of collaborative strategies exhibit a higher level of diversity among neighboring BSs. 

\begin{figure*}[htb]
	\begin{minipage}[t]{0.5\linewidth}
		\includegraphics[width=0.9\textwidth]{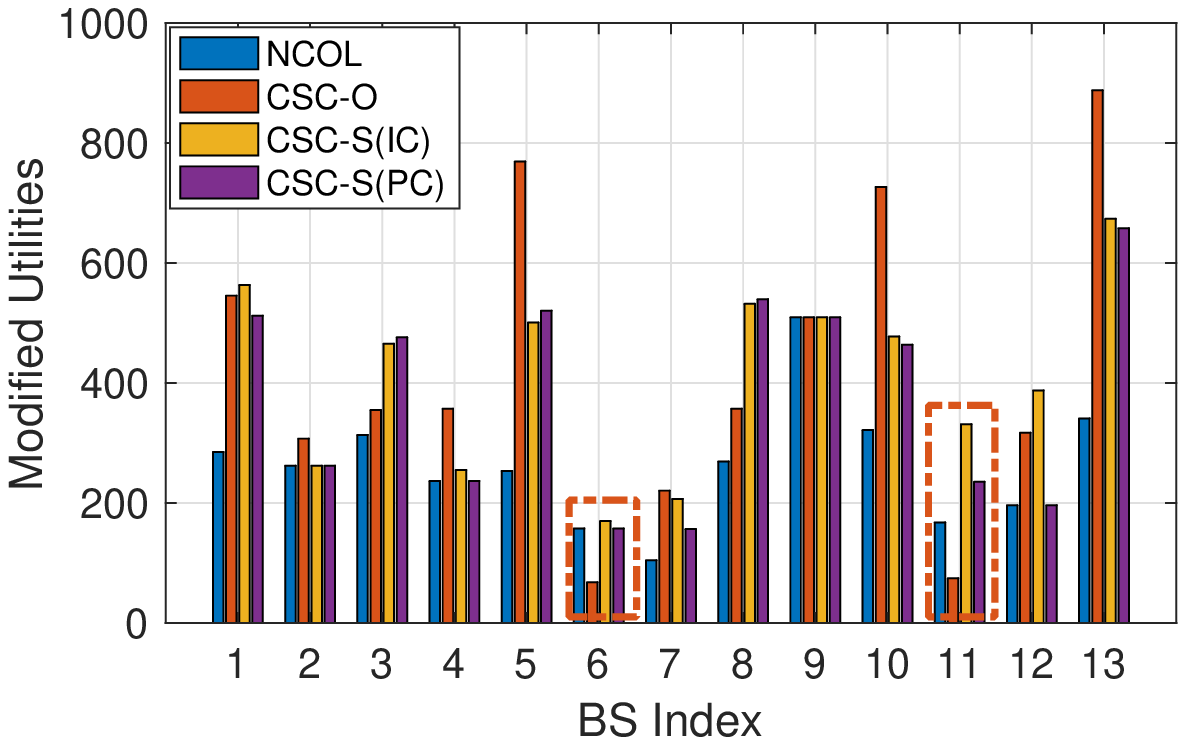}
		\vspace{-0.15 in}
		\caption{Utilities of individual BSs}
		\label{AP_utility}
	\end{minipage}%
	\begin{minipage}[t]{0.5\linewidth}
		\includegraphics[width=0.9\textwidth]{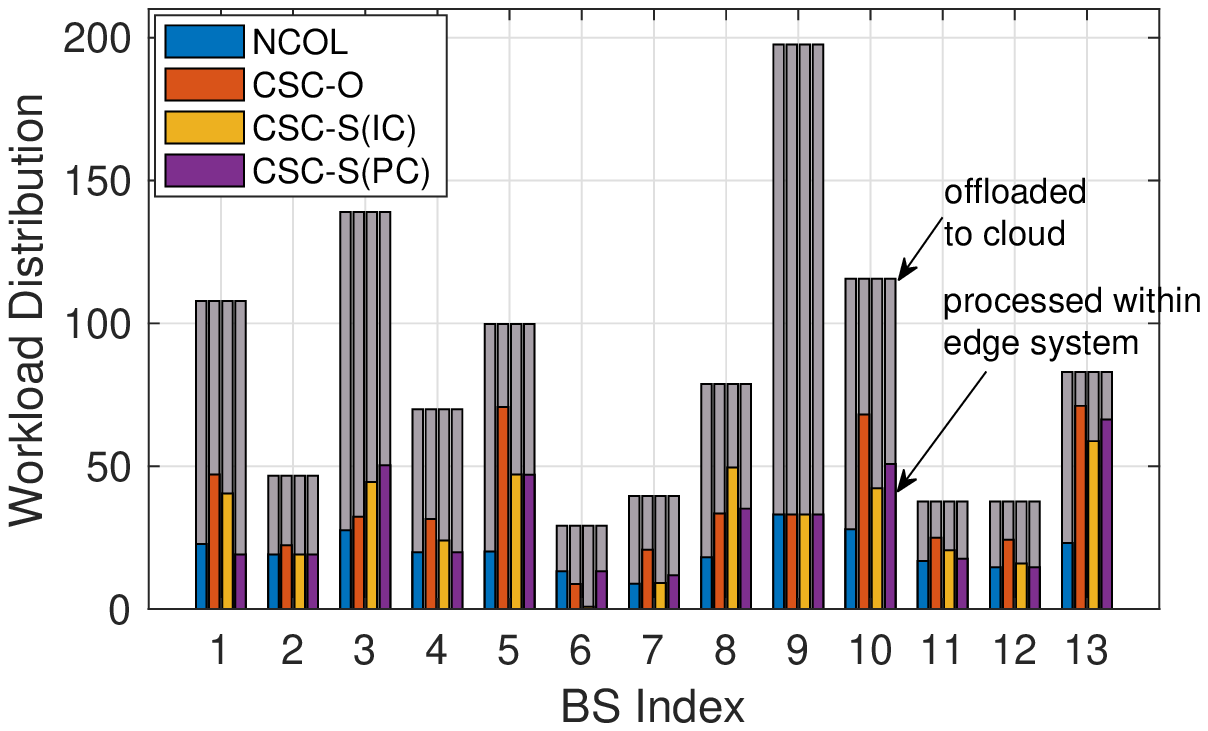}
		\vspace{-0.15 in}
		\caption{Workload distribution}
		\label{work_distribution}
	\end{minipage}%
\vspace{-0.2 in}
\end{figure*}

\vspace{-0.1 in}

\begin{figure}[htb]
	\centering	
    \subfigure[Non-collaborative BS]{\label{a_NCOL}
    	\hspace{-0.2 in}
	\includegraphics[width=0.33\textwidth]{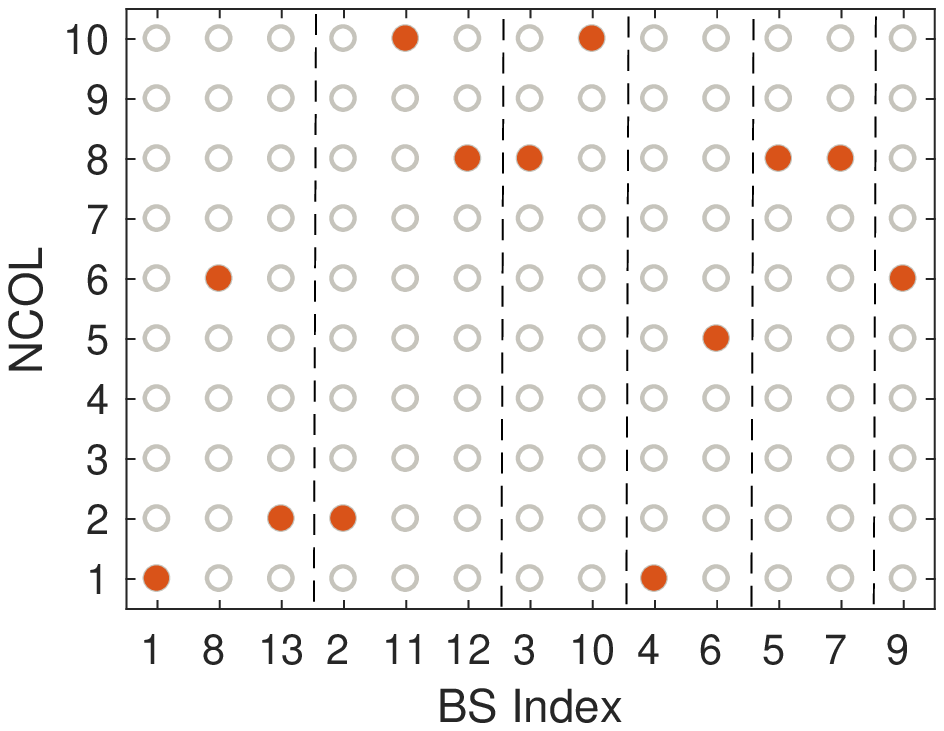}}
	\subfigure[Collaborative obedient BSs]{\label{a_obedient}
		\hspace{-0.2 in}
	\includegraphics[width=0.33\textwidth]{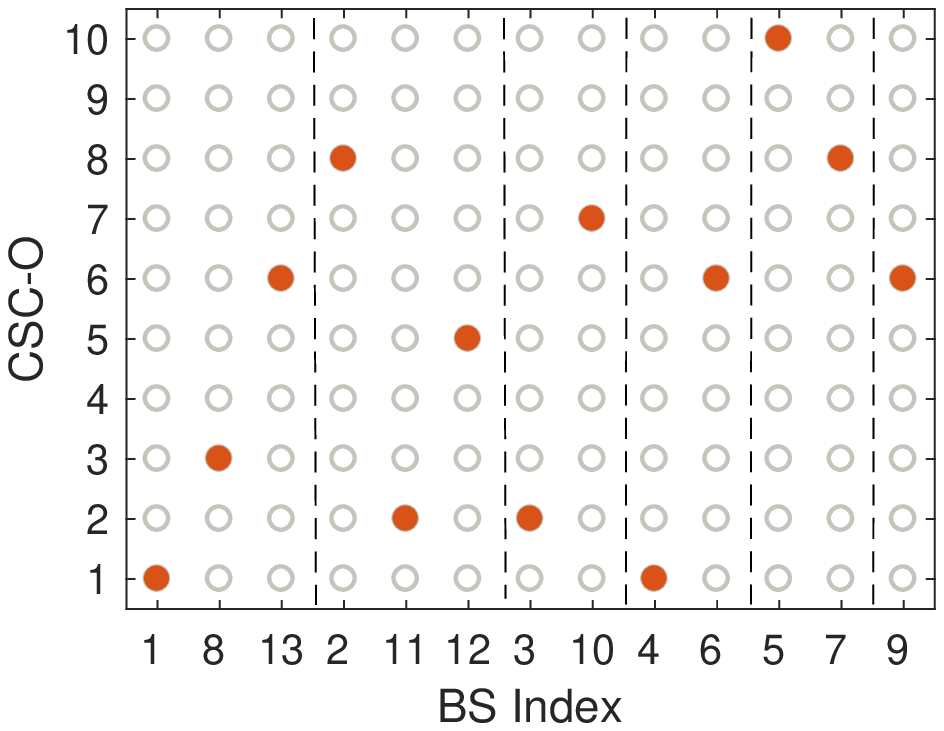}}
	\subfigure[Collaborative strategic BSs]{\label{a_strategic}
		\hspace{-0.2 in}
	\includegraphics[width=0.33\textwidth]{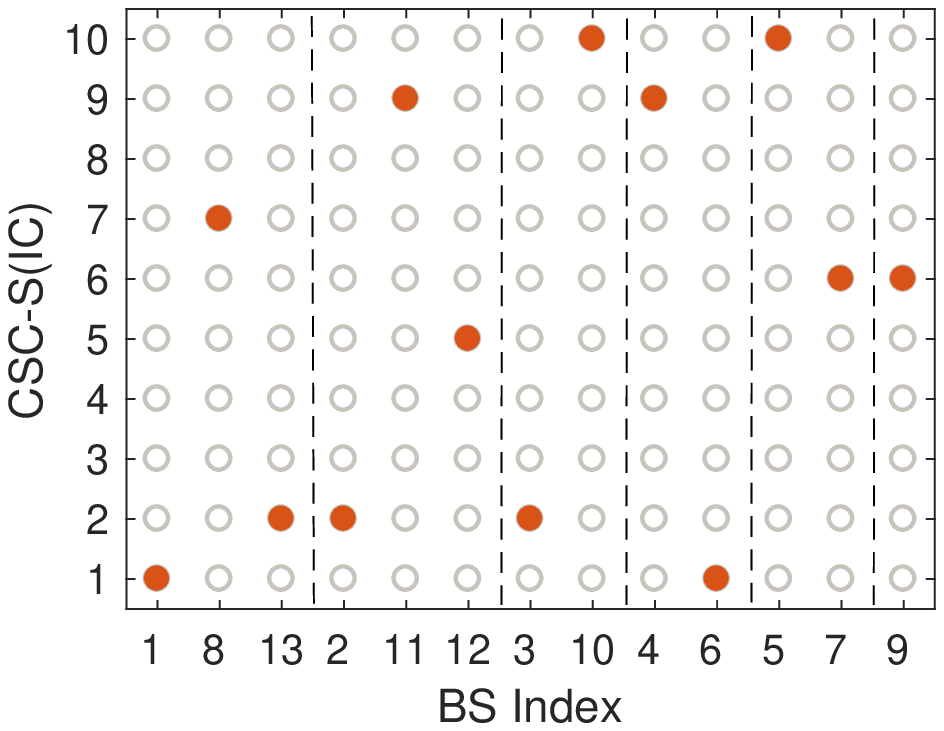}}	
	\caption{Service caching decision}
	\vspace{-0.2 in}
	\label{serv_caching_decision}
\end{figure}

\subsection{Chromatic Parallel Gibbs sampling}
Fig. \ref{chrm_gibbs} compares the convergence processes of CPGS and sequential Gibbs sampler when running CSC-O. It can be observed that CPGS converges much faster than the sequential Gibbs sampler yet reaches to the same optimal system utility. We further present the sampling probability of service caching decisions. Fig. \ref{gibbs_distr_uc} averages the action sampling probabilities of BSs in the first 10 iterations. It can be observed that in this stage the distribution of sampling probabilities are relatively even among potential actions for most BSs. Fig. \ref{gibbs_distr_c} averages the sampling probabilities in the last 10 iterations and we can see that the sampling probabilities become more focused and the optimal service caching decisions are sampled with a high probability, approaching to 1 for most BSs. This implies that the sampling probability has converged to the Gibbs distribution.

\begin{figure*}[htb]
	\begin{minipage}{0.32\linewidth}
	  \centering	
	  \includegraphics[width=1\textwidth]{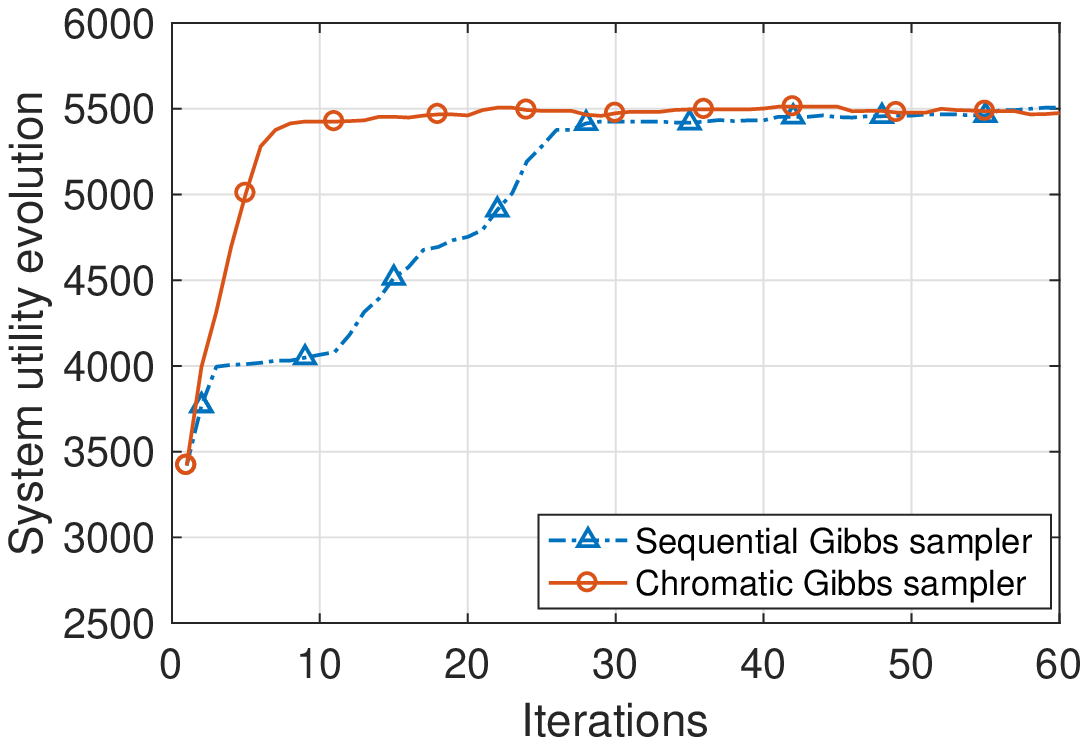}
      \caption{Convergence comparison}
	  \label{chrm_gibbs}
	\end{minipage}%
	\begin{minipage}{0.34\linewidth}
	  \centering	
	  \includegraphics[width=0.9 \textwidth]{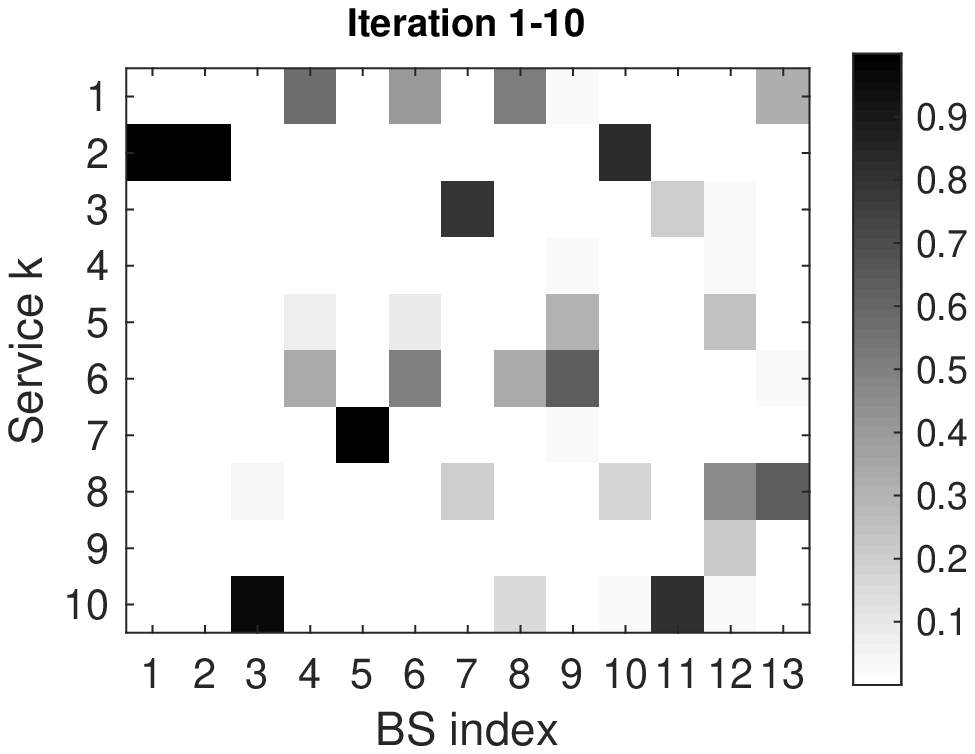}
	  \vspace{-0.1 in}
	  \caption{Unconverged Gibbs distribution}
      \label{gibbs_distr_uc}
    \end{minipage}%
	\begin{minipage}{0.34\linewidth}
	  \centering	
	  \includegraphics[width=0.9\textwidth]{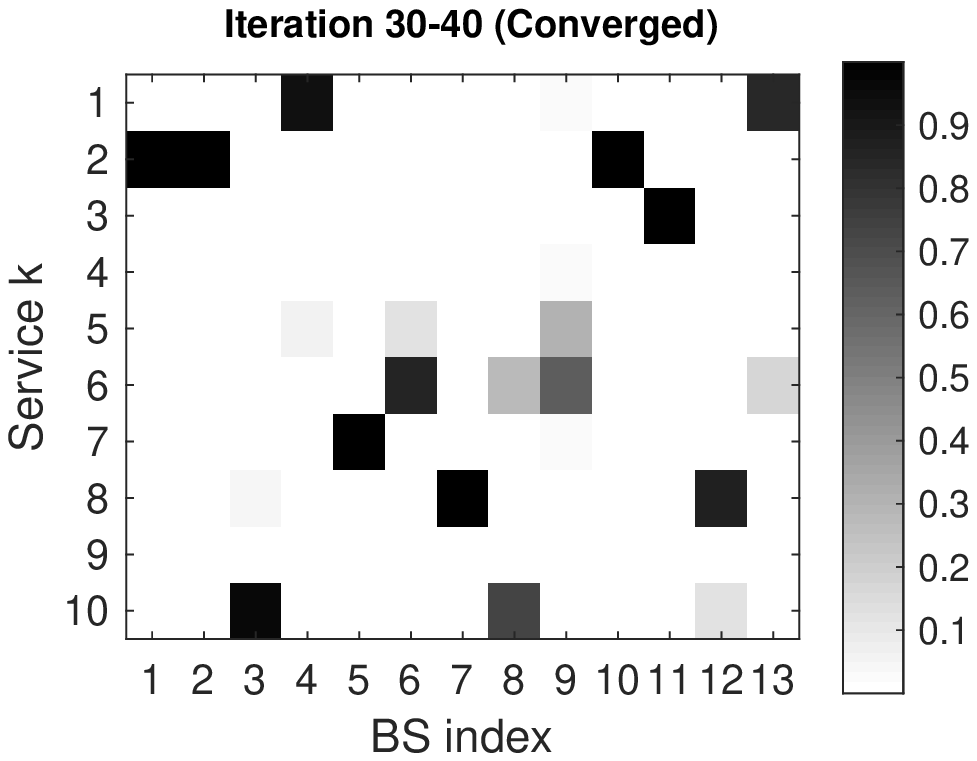}
	  \vspace{-0.1 in}
	  \caption{Converged Gibbs distribution}
      \label{gibbs_distr_c}
    \end{minipage}%
   \vspace{-0.15 in}
\end{figure*}
\vspace{-0.25 in}
%\begin{figure}[htb]
%	\centering	
%	\includegraphics[width=0.45\textwidth]{Figure/chrm_gibbs.eps}
%	\caption{Convergence comparison of sequential and chromatic Gibbs sampler}
%	\label{chrm_gibbs}
%	%\vspace{-0.2in}
%\end{figure}

%\begin{figure}[htb]
%	\centering	
%	\subfigure[Unconverged]{\label{gibbs_distr_uc}
%		\includegraphics[width=0.4\textwidth]{Figure/gibbs_distr_uc.eps}}
%	\subfigure[Converged]{\label{gibbs_distr_c}
%		\includegraphics[width=0.4\textwidth]{Figure/gibbs_distr_c.eps}}
%	\caption{Gibbs distribution}
%	\label{gibbs_distr}
%\end{figure}

\subsection{Coalition formation game}

Fig. \ref{coalitions} presents the generated coalitions by applying CSC-S(IC) and CSC-S(PC). Several points are worth noticing. First, the members in one coalition are not necessarily neighbors. However, collaboration only occurs between neighbors.  For example, in the coalition $\{1,5,7,8,11,12,13\}$ (Fig. \ref{coalitions_PF}), BS 1 and BS 12 are not physically adjacent to each other and therefore no collaboration takes place between them, yet these two BSs belong to the same coalition. Second, a BS may not want to join any coalition. In other words, a BS may want to form an isolated coalition that contains only itself since collaborating with any other BS in the network leads to utility loss. In this particular simulation, we observe that BS 2 and BS 9 separately form isolated coalitions with incentivized collaboration. As a result, the utilities of these two BS stay the same before and after the coalition formation. Third, the BSs in the same coalition tend to cache different types of service in order to accommodate more workload in the edge system. By comparing Fig. \ref{coalitions_PF} and Fig. \ref{coalitions_PWR}, we also see that \textit{Incentivized Collaboration} scheme promotes BSs to form more and larger coalitions compared with \textit{Plain Collaboration}.
\begin{figure}[htb]
	\centering	
	\subfigure[Coalition formation with IC]{\label{coalitions_PF}
		\includegraphics[width=0.4\textwidth]{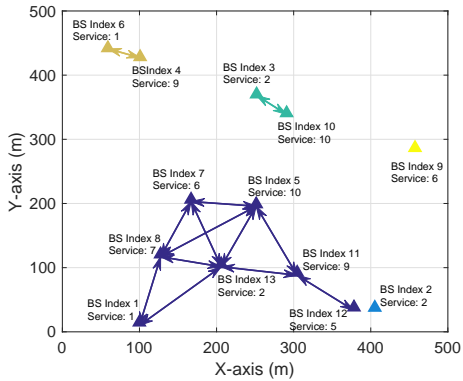}}
	\hspace{0.2 in}
	\subfigure[Coalition formation with PC]{\label{coalitions_PWR}
		\includegraphics[width=0.4\textwidth]{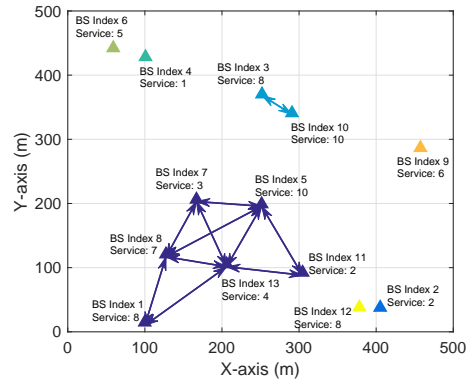}}
	\caption{Coalitions generated by CSC-S}
	\label{coalitions}
	\vspace{-0.25 in}
\end{figure}

During the coalition formation, each iteration is executed by following the Merge/Split operation that aims to find a Pareto-dominant coalition partition than the current partition. Therefore, after each iteration, at least one of the BSs improves its utilities without decreasing the utilities of other BSs. Figure \ref{utility_evolution} shows the system utility evolution of each specific Merge or Split
operation. We see that the system utility is improved with every Merge/Split operation and after only several iterations, the network converges to a stable partition of coalitions. This indicates that in practice, the complexity of running the proposed algorithm is low and hence, it can be easily implemented.

Fig. \ref{payments} show the payments of each SC. It can be verified that the payment of SCs is cleared within each coalition. For instance, in the coalition $\{1,5,7,8,10,11,12,13\}$ with \textit{Incentivized Collaboration}, SCs 1, 7, 11, 12 pay  $(32.5+75.7+94.9+104.2)$ and SCs 5, 8, 13 receive payment $(100.4+ 120.7+86.2)$.

\begin{figure*}[htb]
	\begin{minipage}[t]{0.5\linewidth}
		\centering	
		\includegraphics[width=0.9\textwidth]{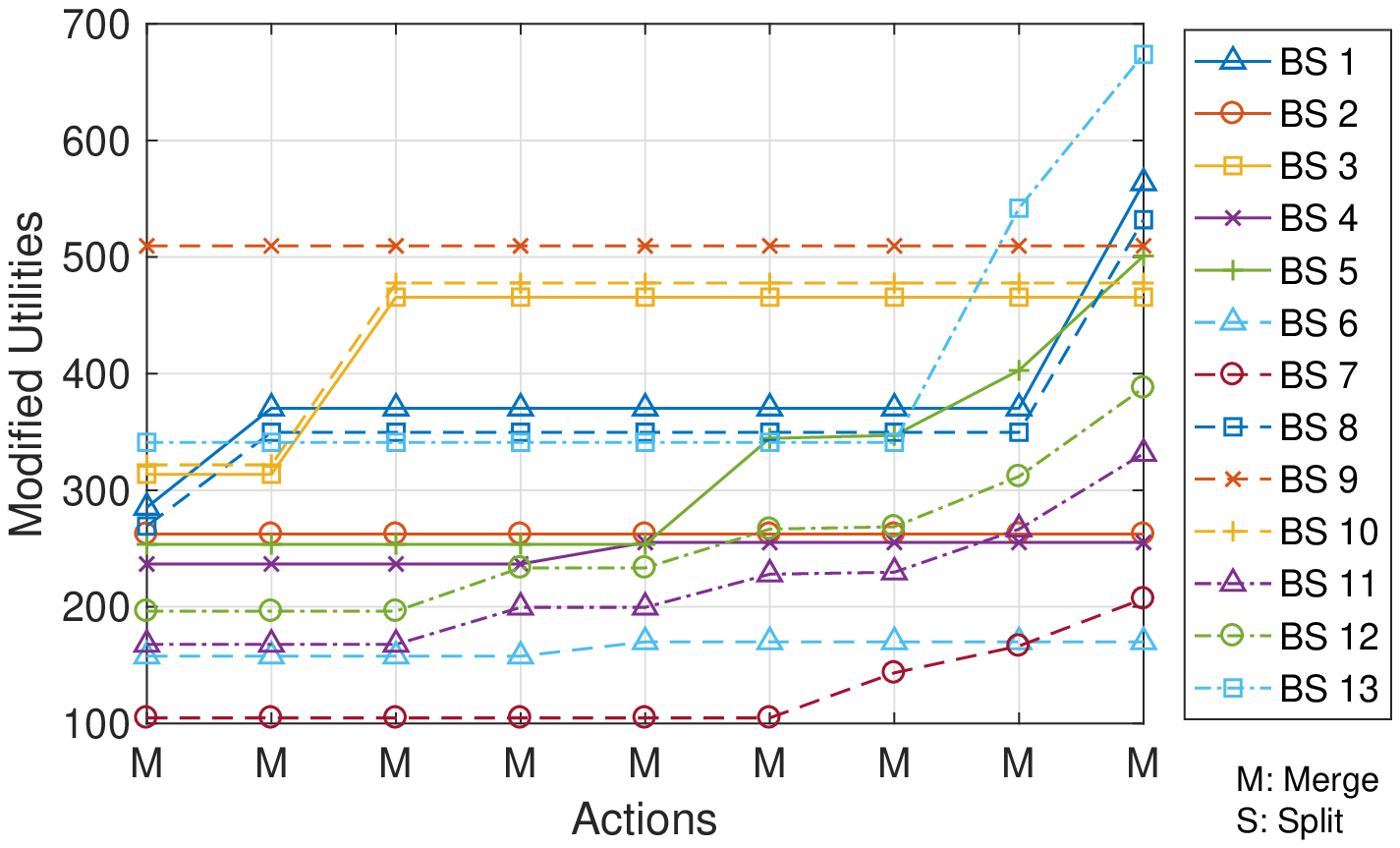}
		\caption{Utility evolution of SCs (IC)}
		\label{utility_evolution}
	\end{minipage}%
	\begin{minipage}[t]{0.5\linewidth}
		\centering	
		\includegraphics[width=1\textwidth]{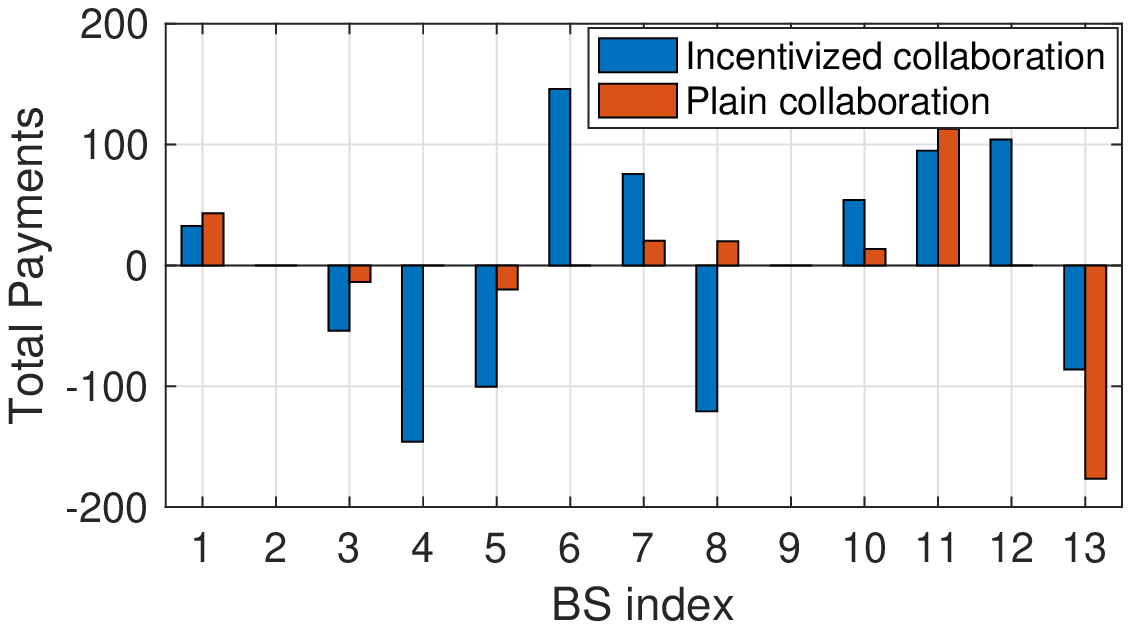}
		\caption{Payments of SCs}
		\label{payments}
	\end{minipage}%
\end{figure*}

\section{Conclusion}
In this paper, we studied collaborative service caching for dense small cell networks. We proposed an efficient decentralized algorithm, called CSC (Collaborative Service Caching), that tailors service caching decisions to collaborative small cell network with obedient (CSC-O) or strategic (CSC-S) BSs. The CSC-O is built on the chromatic parallel Gibbs sampler which tremendously accelerates the convergence of conventional Gibbs sampler yet provides provable optimality performance. The proposed framework is further extended to the SC network consisting of strategic BSs with the assistance of coalitional game. The developed coalition formation algorithm ensures that even selfish BSs are well-motivated to collaborate with other BSs in the coalition via proper incentive schemes. Although our work makes a valuable first step towards optimizing MEC considering service heterogeneity, there are a few limitations in the current model that demand future research effort. For example, user-cell association (load dispatching) decisions can be incorporated in the optimization framework and service demand prediction methods need to be developed for making service caching decisions.

\bibliographystyle{IEEEtran}
\bibliography{refs}
% or


%\appendices
%% you can choose not to have a title for an appendix
%% if you want by leaving the argument blank
%\section{}
%% use section* for acknowledgement
%\section*{Acknowledgment}
%


\ifCLASSOPTIONcaptionsoff
  \newpage
\fi

\end{document}